\documentclass[12pt]{article}

\usepackage[margin=1in]{geometry}
\usepackage{mathpazo}
\usepackage{stmaryrd}
\usepackage[english]{babel}
\usepackage[autostyle, english = american]{csquotes}

\usepackage{tikz}

\usetikzlibrary{
  arrows, 
  automata, 
  shapes, 
  shapes.geometric, 
  positioning,
  calc, 
  chains, 
  decorations.pathmorphing,
  intersections}

\usepackage{hyperref}
\hypersetup{pdfpagemode=UseNone}

\usepackage{amssymb, amsthm, amsmath, dsfont, mathtools}
\usepackage{latexsym}
\usepackage{eepic}
\usepackage{sectsty}

\newtheorem{theorem}{Theorem}
\newtheorem{lemma}[theorem]{Lemma}

\theoremstyle{definition}
\newtheorem{definition}[theorem]{Definition}

\newtheorem{remark}[theorem]{Remark}

\newcommand{\tinyspace}{\mspace{1mu}}

\newcommand{\op}[1]{\operatorname{#1}}
\newcommand{\tr}{\operatorname{Tr}}

\newcommand{\im}{\operatorname{im}}

\newcommand{\abs}[1]{\lvert #1 \rvert}

\newcommand{\ip}[2]{\langle #1 , #2\rangle}

\newcommand{\ket}[1]{\lvert\tinyspace #1 \tinyspace \rangle}
\newcommand{\bra}[1]{\langle\tinyspace #1 \tinyspace \rvert}

\newcommand{\I}{\mathds{1}}

\newcommand{\setft}[1]{\mathrm{#1}}
\newcommand{\Density}{\setft{D}}
\newcommand{\Pos}{\setft{Pos}}

\newcommand{\Proj}{\setft{Proj}}
\newcommand{\Unitary}{\setft{U}}
\newcommand{\Herm}{\setft{Herm}}
\newcommand{\Lin}{\setft{L}}

\newcommand{\complex}{\mathbb{C}}

\newcommand{\integer}{\mathbb{Z}}

\newcommand{\reg}[1]{\mathsf{#1}}

\newcommand\X{\mathcal{X}}
\newcommand\Y{\mathcal{Y}}
\newcommand\Z{\mathcal{Z}}
\newcommand\W{\mathcal{W}}

\newcommand\V{\mathcal{V}}

\renewcommand\S{\mathcal{S}}
\newcommand\T{\mathcal{T}}
\newcommand\K{\mathcal{K}}
\renewcommand\H{\mathcal{H}}

\begin{document}

\title{\bf Revisiting the simulation of quantum Turing machines by quantum
  circuits}

\author{
  {\large Abel Molina \quad and \quad John Watrous}\\[2mm]
  {\it Institute for Quantum Computing and School of Computer Science}\\
  {\it University of Waterloo}}

\date{August 5, 2018}

\maketitle

\begin{abstract}
  Yao (1993) proved that quantum Turing machines and uniformly generated
  quantum circuits are polynomially equivalent computational models:
  $t \geq n$ steps of a quantum Turing machine running on an input of length
  $n$ can be simulated by a uniformly generated family of quantum circuits with
  size quadratic in $t$, and a polynomial-time uniformly generated family of
  quantum circuits can be simulated by a quantum Turing machine running in
  polynomial time.
  We revisit the simulation of quantum Turing machines with uniformly
  generated quantum circuits, which is the more challenging of the two
  simulation tasks, and present a variation on the simulation method
  employed by Yao together with an analysis of it.
  This analysis reveals that the simulation of quantum Turing machines can be
  performed by quantum circuits having depth linear in $t$, rather than
  quadratic depth, and can be extended to variants of quantum Turing machines,
  such as ones having multi-dimensional tapes.
  Our analysis is based on an extension of a method of Arrighi, Nesme, and
  Werner (2011) that allows for the localization of causal unitary evolutions.
\end{abstract}

\section{Introduction}
\label{sec:introduction}

The Turing machine model of computation, proposed by Turing in his landmark
1937 paper \cite{Turing37}, is a cornerstone of computability and complexity
theory: it provides a simple and clean mathematical abstraction of what it
means to compute, upon which a rigorous theory of computation may be
constructed.
It is therefore natural that in the early days of quantum computing,
researchers investigated quantum variants of the Turing machine model, and used
these variants as a foundation from which to study the power of quantum
computing \cite{Deutsch85,DeutschJ92,BernsteinV93,BernsteinV97,Simon94,
  Simon97,Shor94,AdlemanDH97}.

In retrospect, however, it is reasonable to consider the quantum Turing machine
model to be a rather cumbersome model, and not a particularly effective tool
with which to reason that quantum algorithms can be efficiently implemented on
a quantum computer.
The quantum circuit model offers a more usable alternative.
Quantum circuits were first described by Deutsch in 1989 \cite{Deutsch89},
although the now standard acyclic variant of the quantum circuit model was
proposed and investigated a few years later by Yao \cite{Yao93}.
Yao considered the complexity theoretic aspects of quantum circuits, which were
ignored by Deutsch for the most part, proving that (up to a polynomial
overhead) quantum Turing machines and uniformly generated quantum circuits
(i.e., families of quantum circuits that can be efficiently described by
classical Turing machine computations) are equivalent in computational power.
Specifically, Yao proved that $t \geq n$ steps of a quantum Turing machine
running on an input of length $n$ can be simulated by a uniformly generated
family of quantum circuits with size quadratic in $t$, and that a
polynomial-time uniformly generated family of quantum circuits can be simulated
by a quantum Turing machine running in polynomial time.
By the mid- to late-1990s, quantum circuits effectively supplanted quantum
Turing machines as the computational model of choice in the study of quantum
algorithms and complexity theory---a shift made possible by Yao's proof that
the models are equivalent.
The simulation of quantum circuit families by quantum Turing machines is quite
straightforward, and for this reason we will not discuss it further and will
instead focus on the simulation of quantum Turing machines by quantum
circuits.

In this paper we present a variation on Yao's simulation method; the essential
idea behind the simulation we present is the same as Yao's, but the technical
details are somewhat different.
We do not claim that our simulation achieves a quantitative improvement over
Yao's simulation, but we believe nevertheless that there is value in an
alternative simulation and analysis, and also in a discussion that fills in
some of the details absent from Yao's original paper (which appeared only as
an extended abstract in a conference proceedings).
One small advantage of our simulation is that it allows one to essentially
read off an explicit description of the quantum circuits that perform the
simulation from a simple formula, whereas Yao's simulation requires that
one uses linear algebra to solve for a suitable circuit description.
We observe that the simulation of quantum Turing machines by quantum circuits
can be parallelized, resulting in quantum circuits having depth linear in~$t$
rather than depth quadratic in $t$, while still having size quadratic in $t$.
(This is true of our simulation, and although it is not the case for the
precise simulation presented by Yao, it is not difficult to achieve a similar
parallelization by slightly modifying his simulation.)
We also observe that both simulations can be extended to variants of quantum
Turing machines such as ones having multi-dimensional tapes.

Our analysis is based on an extension of a result of Arrighi, Nesme, and Werner
\cite{ArrighiNW11} that allows for the localization of causal unitary
evolutions. 
This extension concerns unitary evolutions that are only causal when restricted
to certain subspaces, and may potentially find other uses, as might also our
introduction of a model of quantum Turing machines with a finite tape loop.

\subsection*{Paper organization}

The remainder of this paper is organized as follows.
First, in Section~\ref{sec:DTM}, we discuss the classic simulation of
deterministic Turing machines by Boolean circuits.
While this simulation cannot be applied directly to quantum Turing machines,
it is useful to refer to it and to view the method as a foundation of the
quantum simulation.
In Section~\ref{sec:QTM} we discuss the quantum Turing machine model at
a formal level, and also introduce the notion of a quantum Turing machine with
a finite tape loop, which simplifies somewhat the study of bounded computations
of quantum Turing machines.
In Section~\ref{sec:localizability} we discuss the localization of causal
unitary evolutions, as described by Arrighi, Nesme, and Werner
\cite{ArrighiNW11}, and prove an extension of their result to unitary
evolutions that are only causal when restricted to certain subspaces.
While it is key to the simulation we consider, this section is completely
independent of the notion of quantum Turing machines, and might potentially be
useful in other contexts.
In Section~\ref{sec:simulation} we present and analyze a simulation of
quantum Turing machines by quantum circuits, compare the simulation with Yao's
original simulation, and briefly discuss how the simulation can be applied to
variants of quantum Turing machines.

\section{The classic Boolean circuit simulation of deterministic Turing
  machines} \label{sec:DTM}

To explain our variation on Yao's method for simulating quantum Turing machine
computations with quantum circuits, and the challenges that both simulations
overcome, it is helpful to recall the classic simulation of deterministic
Turing machines by Boolean circuits.
This discussion also serves as an opportunity to introduce some notation that
will be useful when discussing quantum circuit simulations of quantum Turing
machines.
This classical simulation method, variants of which appear in
\cite{Savage72, PippengerF79} and as standard material in textbooks on
computational complexity, can informally be described as a ``solid state''
implementation of a Turing machine.

For the sake of simplicity, we will assume that the deterministic Turing
machine to be simulated has state set $Q = \{1,\ldots,m\}$ and tape alphabet
$\Gamma = \{0,\ldots,k-1\}$, where the tape symbol 0 represents the blank
symbol.
The tape is assumed to be two-way infinite, with squares indexed by the set of
integers $\integer$.
The computation begins with the tape head scanning the tape square indexed by
0, with the input string written in the squares indexed by $1,\ldots,n$ and all
other tape squares containing the blank symbol, and with the starting
state~$1$.
The evolution of the Turing machine is specified by a transition function
\begin{equation}
  \delta: Q \times \Gamma \rightarrow Q \times \Gamma \times \{-1,+1\};
\end{equation}
if the machine is in the state $p\in Q$, the tape head is scanning a square
that holds the symbol $a\in\Gamma$, and it is the case that
\begin{equation}
  \delta(p,a) = (q, b, D),
\end{equation}
then in one step the machine will change state to $q$, overwrite the currently
scanned tape square with the symbol $b$, and move the tape head in the
direction $D$ (where $D = -1$ indicates a movement to the left and $D = +1$
indicates a movement to the right).

Suppose that $t$ steps of the Turing machine's computation are to be simulated,
and assume that the length of the input string satisfies $n \leq t$.
Note that it is not possible for the tape head to leave the region of the tape
indexed by the set $\{-t,\ldots,t\}$, and no tape square outside of this region
will ever store a non-blank symbol within these $t$ steps.
For each of the $2t+1$ tape squares indexed by $\{-t,\ldots,t\}$, one
imagines two registers: one register stores information that indicates whether
or not the tape head is currently scanning this tape square, and if it is, the
current state of the Turing machine, while the other register stores the
symbol that is currently written in the corresponding tape square.
More precisely, we define registers
\begin{equation}
  \label{eq:registers-S-and-T}
  \reg{S}_{-t},\ldots,\reg{S}_{t}
  \quad\text{and}\quad
  \reg{T}_{-t},\ldots,\reg{T}_{t},
\end{equation}
where each register $\reg{S}_i$ holds an element of the set $\{0,\ldots,m\}$
and each register $\reg{T}_i$ holds an element of the set $\{0,\ldots,k-1\}$.
If $\reg{S}_i$ holds $0$, then the tape head is not positioned over the square
indexed by $i$, while if $\reg{S}_i$ holds $p \in \{1,\ldots,m\}$, then the
tape square is positioned over the tape square indexed by $i$ and the current
state of the Turing machine is $p$.
In either case, the symbol written in the tape square indexed by $i$ is stored
in $\reg{T}_i$.
The simulation begins with a pre-processing step that initializes these
registers appropriately for a given input string.

To simulate one step of the Turing machine's computation, each of the
registers \eqref{eq:registers-S-and-T} is simultaneously updated.
The contents of the pair $(\reg{S}_i,\reg{T}_i)$ after being updated depend
only on the contents of the registers $(\reg{S}_{i-1},\reg{T}_{i-1})$,
$(\reg{S}_{i},\reg{T}_{i})$, and $(\reg{S}_{i+1},\reg{T}_{i+1})$ prior to the
update, as is suggested by Figure~\ref{fig:local-classical-transitions}.
This local dependence is enabled by the fact that the Turing machine's state is
stored locally in the register $\reg{S}_j$ that corresponds to the tape head
location $j$, along with the fact that the tape head cannot move more than one
square on each computation step.
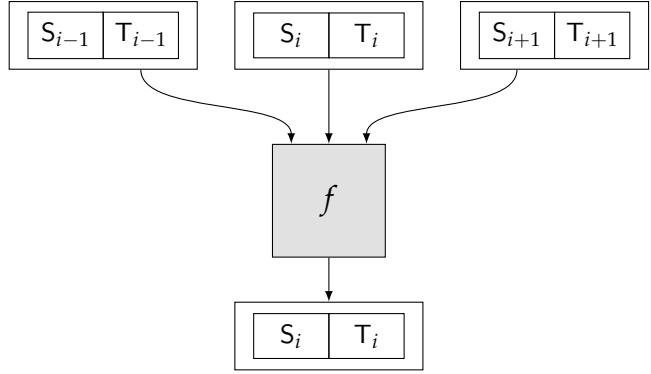
\begin{figure}[!t]
  \begin{center}
    \vspace*{4mm}
    \begin{tikzpicture}[>=latex]
      \tikzstyle{block}=[draw, minimum width=25mm, minimum height=9mm]
      \tikzstyle{circuit}=[draw, minimum width=15mm, minimum height=15mm,
        fill=black!12]
      \tikzstyle{iblock}=[minimum height=9mm]
      \tikzstyle{icircuit}=[minimum height=15mm]
      \tikzstyle{register}=[draw, minimum width=10mm, minimum height=4mm,
        font=\fontsize{10}{0}\selectfont]
      
      \node[block] (A1) at (-4.5,6) {};
      \node[block] (A2) at (-1.5,6) {};
      \node[block] (A3) at (1.5,6) {};
      \node[block] (B2) at (-1.5,2) {};

      \node[circuit,yshift=-22mm] (cA2) at (A2) {$f$};
      
      \node[register,xshift=-5mm] at (A1) {$\reg{S}_{i-1}$};
      \node[register,xshift=5mm] at (A1) {$\reg{T}_{i-1}$};
      
      \node[register,xshift=-5mm] at (A2) {$\reg{S}_i$};
      \node[register,xshift=5mm] at (A2) {$\reg{T}_i$};
      
      \node[register,xshift=-5mm] at (A3) {$\reg{S}_{i+1}$};
      \node[register,xshift=5mm] at (A3) {$\reg{T}_{i+1}$};
      
      \node[register,xshift=-5mm] at (B2) {$\reg{S}_i$};
      \node[register,xshift=5mm] at (B2) {$\reg{T}_i$};
      
      \draw[->] ([xshift=5mm]A1.south)
      .. controls +(down:6mm) and +(up:6mm) ..
      ([xshift=-5mm]cA2.north);
      
      \draw[->] ([xshift=0mm]A2.south) 
      .. controls +(down:6mm) and +(up:6mm) ..
      ([xshift=0]cA2.north);
      
      \draw[->] ([xshift=-5mm]A3.south) 
      .. controls +(down:6mm) and +(up:6mm) ..
      ([xshift=5mm]cA2.north);
      
      \draw[->] (cA2.south) -- (B2.north);
      
    \end{tikzpicture}
  \end{center}
  \caption{The information corresponding to cell $i$ in the Turing machine is a
    function of the information in the previous time step for cells
    $\{i-1, i, i+1\}$.
    The shaded box labeled $f$ represents the function that determines the
    updated contents of $(\reg{S}_i,\reg{T}_i)$ given the contents of
    $(\reg{S}_{i-1},\reg{T}_{i-1})$, $(\reg{S}_i,\reg{T}_i)$, and
    $(\reg{S}_{i+1},\reg{T}_{i+1})$ prior to being updated.}
  \label{fig:local-classical-transitions}
\end{figure}
For example, if $(\reg{S}_i,\reg{T}_i)$ stores $(0,c)$ for some
$c\in\{0,\ldots,k-1\}$, then after being updated, $\reg{T}_i$ will continue
to store $c$, as the absence of the tape head at square $i$ prevents the
tape contents of square $i$ from changing.
However, after being updated, $\reg{S}_i$ might potentially contain any element
of $\{0,\ldots,m\}$; for example, if $(\reg{S}_{i+1},\reg{T}_{i+1})$ stores
$(p,a)$ and $\delta(p,a) = (q, b, -1)$, then the updated
contents of $\reg{S}_i$ will become $q$, indicating that the tape head has
moved over square $i$ and that the current state has become~$q$.
In the same situation the updated contents of $(\reg{S}_{i+1},\reg{T}_{i+1})$
will become $(0,b)$.

One could of course give an explicit description of this update rule, for a
given transition function $\delta$, but doing so for the sake of this
discussion is not particularly helpful.
Instead, it suffices to observe that there exists a function
\begin{equation}
  f: \bigl( \{0,\ldots,m\}\times\{0,\ldots,k-1\}\bigr)^3
  \rightarrow \{0,\ldots,m\}\times\{0,\ldots,k-1\},
\end{equation}
also depicted in Figure~\ref{fig:local-classical-transitions}, that describes
the update rule.
Figure~\ref{fig:circuit-one-step} illustrates the update being applied
simultaneously to every pair of registers.

It should be noted that for inputs of the form
\begin{equation}
  ((q_1,a_1),(q_2,a_2),(q_3,a_3))
\end{equation}
where two or more of the elements $q_1,q_2,q_3$ are contained in the set
$\{1,\ldots,m\}$, the output of $f$ may be defined arbitrarily; such an input
could only arise from a Turing machine configuration having two or more tape
heads, which never happens in a valid Turing machine computation.
Also note that a special case must be made for the register pairs in the
leftmost column (i.e., $(\reg{S}_{-t},\reg{T}_{-t})$) and the rightmost column
(i.e., $(\reg{S}_{t},\reg{T}_{t})$) as there are no register pairs to the left
or right, respectively, to feed into the function $f$ that determines how these
registers update.
However, the missing inputs will always be $(0,0)$, representing the absence of
the tape head and a blank symbol stored on the tape, and so they can be
``hard-coded'' into the corresponding copies of the function $f$.

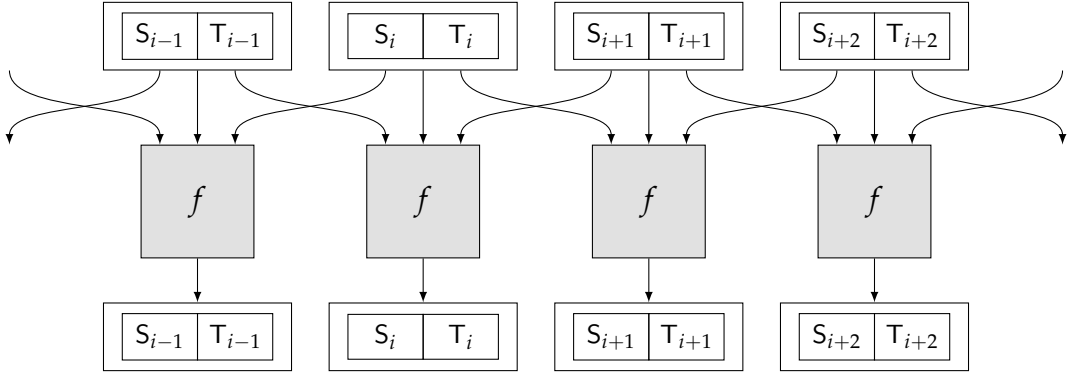
\begin{figure}[!t]
  \begin{center}
    \begin{tikzpicture}[>=latex]
      \tikzstyle{block}=[draw, minimum width=25mm, minimum height=9mm]
      \tikzstyle{circuit}=[draw, minimum width=15mm, minimum height=15mm,
        fill=black!12]
      \tikzstyle{iblock}=[minimum height=9mm]
      \tikzstyle{icircuit}=[minimum height=15mm]
      \tikzstyle{register}=[draw, minimum width=10mm, minimum height=4mm,
        font=\fontsize{10}{0}\selectfont]
      
      \node[iblock] (A0) at (-7.5,6) {};
      \node[block] (A1) at (-4.5,6) {};
      \node[block] (A2) at (-1.5,6) {};
      \node[block] (A3) at (1.5,6) {};
      \node[block] (A4) at (4.5,6) {};
      \node[iblock] (A5) at (7.5,6) {};
      
      \node[iblock] (B0) at (-7.5,2) {};
      \node[block] (B1) at (-4.5,2) {};
      \node[block] (B2) at (-1.5,2) {};
      \node[block] (B3) at (1.5,2) {};
      \node[block] (B4) at (4.5,2) {};
      \node[iblock] (B5) at (7.5,2) {};
      
      \node[icircuit,yshift=-22mm] (cA0) at (A0) {};
      \node[circuit,yshift=-22mm] (cA1) at (A1) {$f$};
      \node[circuit,yshift=-22mm] (cA2) at (A2) {$f$};
      \node[circuit,yshift=-22mm] (cA3) at (A3) {$f$};
      \node[circuit,yshift=-22mm] (cA4) at (A4) {$f$};
      \node[icircuit,yshift=-22mm] (cA5) at (A5) {};
      
      \node[register,xshift=-5mm] at (A1) {$\reg{S}_{i-1}$};
      \node[register,xshift=5mm] at (A1) {$\reg{T}_{i-1}$};
      
      \node[register,xshift=-5mm] at (A2) {$\reg{S}_i$};
      \node[register,xshift=5mm] at (A2) {$\reg{T}_i$};
      
      \node[register,xshift=-5mm] at (A3) {$\reg{S}_{i+1}$};
      \node[register,xshift=5mm] at (A3) {$\reg{T}_{i+1}$};
      
      \node[register,xshift=-5mm] at (A4) {$\reg{S}_{i+2}$};
      \node[register,xshift=5mm] at (A4) {$\reg{T}_{i+2}$};
      
      \node[register,xshift=-5mm] at (B1) {$\reg{S}_{i-1}$};
      \node[register,xshift=5mm] at (B1) {$\reg{T}_{i-1}$};
      
      \node[register,xshift=-5mm] at (B2) {$\reg{S}_i$};
      \node[register,xshift=5mm] at (B2) {$\reg{T}_i$};
      
      \node[register,xshift=-5mm] at (B3) {$\reg{S}_{i+1}$};
      \node[register,xshift=5mm] at (B3) {$\reg{T}_{i+1}$};
      
      \node[register,xshift=-5mm] at (B4) {$\reg{S}_{i+2}$};
      \node[register,xshift=5mm] at (B4) {$\reg{T}_{i+2}$};
      
      \draw[->] ([xshift=-5mm]A1.south)
      .. controls +(down:6mm) and +(up:6mm) ..
      ([xshift=5mm]cA0.north);
      
      \draw[->] ([xshift=5mm]A4.south)
      .. controls +(down:6mm) and +(up:6mm) ..
      ([xshift=-5mm]cA5.north);
      
      \draw[->] ([xshift=5mm]A0.south)
      .. controls +(down:6mm) and +(up:6mm) ..
      ([xshift=-5mm]cA1.north);
      
      \draw[->] ([xshift=0mm]A1.south) 
      .. controls +(down:6mm) and +(up:6mm) ..
      ([xshift=0]cA1.north);
      
      \draw[->] ([xshift=-5mm]A2.south) 
      .. controls +(down:6mm) and +(up:6mm) ..
      ([xshift=5mm]cA1.north);
      
      \draw[->] (cA1.south) -- (B1.north);
      
      \draw[->] ([xshift=5mm]A1.south)
      .. controls +(down:6mm) and +(up:6mm) ..
      ([xshift=-5mm]cA2.north);
      
      \draw[->] ([xshift=0mm]A2.south) 
      .. controls +(down:6mm) and +(up:6mm) ..
      ([xshift=0]cA2.north);
      
      \draw[->] ([xshift=-5mm]A3.south) 
      .. controls +(down:6mm) and +(up:6mm) ..
      ([xshift=5mm]cA2.north);
      
      \draw[->] (cA2.south) -- (B2.north);
      
      \draw[->] ([xshift=5mm]A2.south)
      .. controls +(down:6mm) and +(up:6mm) ..
      ([xshift=-5mm]cA3.north);
      
      \draw[->] ([xshift=0mm]A3.south) 
      .. controls +(down:6mm) and +(up:6mm) ..
      ([xshift=0]cA3.north);
      
      \draw[->] ([xshift=-5mm]A4.south) 
      .. controls +(down:6mm) and +(up:6mm) ..
      ([xshift=5mm]cA3.north);
      
      \draw[->] (cA3.south) -- (B3.north);
      
      \draw[->] ([xshift=5mm]A3.south)
      .. controls +(down:6mm) and +(up:6mm) ..
      ([xshift=-5mm]cA4.north);
      
      \draw[->] ([xshift=0mm]A4.south) 
      .. controls +(down:6mm) and +(up:6mm) ..
      ([xshift=0]cA4.north);
      
      \draw[->] ([xshift=-5mm]A5.south) 
      .. controls +(down:6mm) and +(up:6mm) ..
      ([xshift=5mm]cA4.north);
     
      \draw[->] (cA4.south) -- (B4.north);
      
    \end{tikzpicture}
  \end{center}
  \caption{To simulate one step of a Turing machine computation, each pair of
    registers is simultaneously updated.
    The same function $f$ describes the update for each pair of registers.}
  \label{fig:circuit-one-step}
\end{figure}

To simulate $t$ steps of the Turing machine computation, one envisions
a network consisting of $t+1$ rows, each having the $2t+1$ register pairs
$(\reg{S}_{-t},\reg{T}_{-t}),\ldots, (\reg{S}_{t},\reg{T}_{t})$, as suggested
by Figure~\ref{fig:circuit-pattern}.
The registers in the top row (which corresponds to time $0$) are initialized
by the pre-processing step suggested previously, so that collectively they
describe the initial configuration of the Turing machine on the input string of
interest, and their updates are performed in the indicated pattern.
Subsequent rows of register pairs will then collectively describe the
configuration of the Turing machine on subsequent computation steps, and in
particular the bottom row will describe the configuration after $t$ steps.
A final post-processing step may be appended so that a description of the final
Turing machine configuration is produced that conforms to some alternative
encoding scheme, if that is desired.

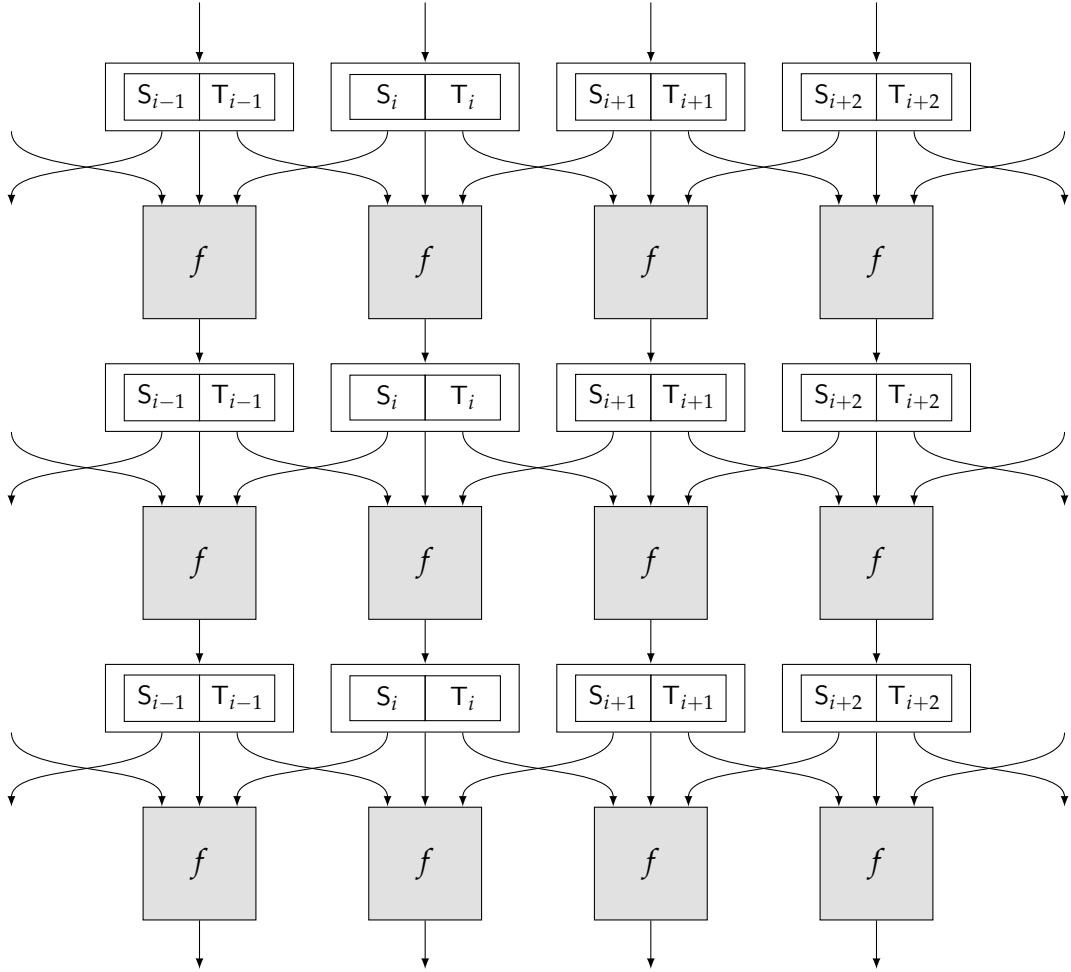
\begin{figure}[!t]
  \begin{center}
    \begin{tikzpicture}[>=latex]
      \tikzstyle{block}=[draw, minimum width=25mm, minimum height=9mm]
      \tikzstyle{circuit}=[draw, minimum width=15mm, minimum height=15mm,
        fill=black!12]
      \tikzstyle{iblock}=[minimum height=9mm]
      \tikzstyle{icircuit}=[minimum height=15mm]
      \tikzstyle{register}=[draw, minimum width=10mm, minimum height=4mm,
        font=\fontsize{10}{0}\selectfont]
      
      \node[iblock] (A0) at (-7.5,6) {};
      \node[block] (A1) at (-4.5,6) {};
      \node[block] (A2) at (-1.5,6) {};
      \node[block] (A3) at (1.5,6) {};
      \node[block] (A4) at (4.5,6) {};
      \node[iblock] (A5) at (7.5,6) {};
      
      \node[iblock] (B0) at (-7.5,2) {};
      \node[block] (B1) at (-4.5,2) {};
      \node[block] (B2) at (-1.5,2) {};
      \node[block] (B3) at (1.5,2) {};
      \node[block] (B4) at (4.5,2) {};
      \node[iblock] (B5) at (7.5,2) {};
      
      \node[iblock] (C0) at (-7.5,-2) {};
      \node[block] (C1) at (-4.5,-2) {};
      \node[block] (C2) at (-1.5,-2) {};
      \node[block] (C3) at (1.5,-2) {};
      \node[block] (C4) at (4.5,-2) {};
      \node[iblock] (C5) at (7.5,-2) {};
      
      \node[minimum height=1mm] (D1) at (-4.5,-5.75) {};
      \node[minimum height=1mm] (D2) at (-1.5,-5.75) {};
      \node[minimum height=1mm] (D3) at (1.5,-5.75) {};
      \node[minimum height=1mm] (D4) at (4.5,-5.75) {};
      
      \node[icircuit,yshift=-22mm] (cA0) at (A0) {};
      \node[circuit,yshift=-22mm] (cA1) at (A1) {$f$};
      \node[circuit,yshift=-22mm] (cA2) at (A2) {$f$};
      \node[circuit,yshift=-22mm] (cA3) at (A3) {$f$};
      \node[circuit,yshift=-22mm] (cA4) at (A4) {$f$};
      \node[icircuit,yshift=-22mm] (cA5) at (A5) {};
      
      \node[icircuit,yshift=-22mm] (cB0) at (B0) {};
      \node[circuit,yshift=-22mm] (cB1) at (B1) {$f$};
      \node[circuit,yshift=-22mm] (cB2) at (B2) {$f$};
      \node[circuit,yshift=-22mm] (cB3) at (B3) {$f$};
      \node[circuit,yshift=-22mm] (cB4) at (B4) {$f$};
      \node[icircuit,yshift=-22mm] (cB5) at (B5) {};
      
      \node[icircuit,yshift=-22mm] (cC0) at (C0) {};
      \node[circuit,yshift=-22mm] (cC1) at (C1) {$f$};
      \node[circuit,yshift=-22mm] (cC2) at (C2) {$f$};
      \node[circuit,yshift=-22mm] (cC3) at (C3) {$f$};
      \node[circuit,yshift=-22mm] (cC4) at (C4) {$f$};
      \node[icircuit,yshift=-22mm] (cC5) at (C5) {};
      
      \node[register,xshift=-5mm] at (A1) {$\reg{S}_{i-1}$};
      \node[register,xshift=5mm] at (A1) {$\reg{T}_{i-1}$};
      
      \node[register,xshift=-5mm] at (A2) {$\reg{S}_i$};
      \node[register,xshift=5mm] at (A2) {$\reg{T}_i$};
      
      \node[register,xshift=-5mm] at (A3) {$\reg{S}_{i+1}$};
      \node[register,xshift=5mm] at (A3) {$\reg{T}_{i+1}$};
      
      \node[register,xshift=-5mm] at (A4) {$\reg{S}_{i+2}$};
      \node[register,xshift=5mm] at (A4) {$\reg{T}_{i+2}$};
      
      \node[register,xshift=-5mm] at (B1) {$\reg{S}_{i-1}$};
      \node[register,xshift=5mm] at (B1) {$\reg{T}_{i-1}$};
      
      \node[register,xshift=-5mm] at (B2) {$\reg{S}_i$};
      \node[register,xshift=5mm] at (B2) {$\reg{T}_i$};
      
      \node[register,xshift=-5mm] at (B3) {$\reg{S}_{i+1}$};
      \node[register,xshift=5mm] at (B3) {$\reg{T}_{i+1}$};
      
      \node[register,xshift=-5mm] at (B4) {$\reg{S}_{i+2}$};
      \node[register,xshift=5mm] at (B4) {$\reg{T}_{i+2}$};
      
      \node[register,xshift=-5mm] at (C1) {$\reg{S}_{i-1}$};
      \node[register,xshift=5mm] at (C1) {$\reg{T}_{i-1}$};
      
      \node[register,xshift=-5mm] at (C2) {$\reg{S}_i$};
      \node[register,xshift=5mm] at (C2) {$\reg{T}_i$};
      
      \node[register,xshift=-5mm] at (C3) {$\reg{S}_{i+1}$};
      \node[register,xshift=5mm] at (C3) {$\reg{T}_{i+1}$};
      
      \node[register,xshift=-5mm] at (C4) {$\reg{S}_{i+2}$};
      \node[register,xshift=5mm] at (C4) {$\reg{T}_{i+2}$};
      
      \draw[->] ([yshift=8mm]A1.north) -- (A1.north);
      \draw[->] ([yshift=8mm]A2.north) -- (A2.north);
      \draw[->] ([yshift=8mm]A3.north) -- (A3.north);
      \draw[->] ([yshift=8mm]A4.north) -- (A4.north);
      
      \draw[->] ([xshift=-5mm]A1.south)
      .. controls +(down:6mm) and +(up:6mm) ..
      ([xshift=5mm]cA0.north);
      
      \draw[->] ([xshift=-5mm]B1.south)
      .. controls +(down:6mm) and +(up:6mm) ..
      ([xshift=5mm]cB0.north);
      
      \draw[->] ([xshift=-5mm]C1.south)
      .. controls +(down:6mm) and +(up:6mm) ..
      ([xshift=5mm]cC0.north);
      
      \draw[->] ([xshift=5mm]A4.south)
      .. controls +(down:6mm) and +(up:6mm) ..
      ([xshift=-5mm]cA5.north);
      
      \draw[->] ([xshift=5mm]B4.south)
      .. controls +(down:6mm) and +(up:6mm) ..
      ([xshift=-5mm]cB5.north);
      
      \draw[->] ([xshift=5mm]C4.south)
      .. controls +(down:6mm) and +(up:6mm) ..
      ([xshift=-5mm]cC5.north);
      
      \draw[->] ([xshift=5mm]A0.south)
      .. controls +(down:6mm) and +(up:6mm) ..
      ([xshift=-5mm]cA1.north);
      
      \draw[->] ([xshift=0mm]A1.south) 
      .. controls +(down:6mm) and +(up:6mm) ..
      ([xshift=0]cA1.north);
      
      \draw[->] ([xshift=-5mm]A2.south) 
      .. controls +(down:6mm) and +(up:6mm) ..
      ([xshift=5mm]cA1.north);
      
      \draw[->] (cA1.south) -- (B1.north);
      
      \draw[->] ([xshift=5mm]B0.south)
      .. controls +(down:6mm) and +(up:6mm) ..
      ([xshift=-5mm]cB1.north);
      
      \draw[->] ([xshift=0mm]B1.south) 
      .. controls +(down:6mm) and +(up:6mm) ..
      ([xshift=0]cB1.north);
      
      \draw[->] ([xshift=-5mm]B2.south) 
      .. controls +(down:6mm) and +(up:6mm) ..
      ([xshift=5mm]cB1.north);
      
      \draw[->] (cB1.south) -- (C1.north);
      
      \draw[->] ([xshift=5mm]C0.south)
      .. controls +(down:6mm) and +(up:6mm) ..
      ([xshift=-5mm]cC1.north);
      
      \draw[->] ([xshift=0mm]C1.south) 
      .. controls +(down:6mm) and +(up:6mm) ..
      ([xshift=0]cC1.north);
      
      \draw[->] ([xshift=-5mm]C2.south) 
      .. controls +(down:6mm) and +(up:6mm) ..
      ([xshift=5mm]cC1.north);
      
      \draw[->] (cC1.south) -- (D1.north);
      
      \draw[->] ([xshift=5mm]A1.south)
      .. controls +(down:6mm) and +(up:6mm) ..
      ([xshift=-5mm]cA2.north);
      
      \draw[->] ([xshift=0mm]A2.south) 
      .. controls +(down:6mm) and +(up:6mm) ..
      ([xshift=0]cA2.north);
      
      \draw[->] ([xshift=-5mm]A3.south) 
      .. controls +(down:6mm) and +(up:6mm) ..
      ([xshift=5mm]cA2.north);
      
      \draw[->] (cA2.south) -- (B2.north);
      
      \draw[->] ([xshift=5mm]B1.south)
      .. controls +(down:6mm) and +(up:6mm) ..
      ([xshift=-5mm]cB2.north);
      
      \draw[->] ([xshift=0mm]B2.south) 
      .. controls +(down:6mm) and +(up:6mm) ..
      ([xshift=0]cB2.north);
      
      \draw[->] ([xshift=-5mm]B3.south) 
      .. controls +(down:6mm) and +(up:6mm) ..
      ([xshift=5mm]cB2.north);
      
      \draw[->] (cB2.south) -- (C2.north);
      
      \draw[->] ([xshift=5mm]C1.south)
      .. controls +(down:6mm) and +(up:6mm) ..
      ([xshift=-5mm]cC2.north);
      
      \draw[->] ([xshift=0mm]C2.south) 
      .. controls +(down:6mm) and +(up:6mm) ..
      ([xshift=0]cC2.north);
      
      \draw[->] ([xshift=-5mm]C3.south) 
      .. controls +(down:6mm) and +(up:6mm) ..
      ([xshift=5mm]cC2.north);
      
      \draw[->] (cC2.south) -- (D2.north);
      
      \draw[->] ([xshift=5mm]A2.south)
      .. controls +(down:6mm) and +(up:6mm) ..
      ([xshift=-5mm]cA3.north);
      
      \draw[->] ([xshift=0mm]A3.south) 
      .. controls +(down:6mm) and +(up:6mm) ..
      ([xshift=0]cA3.north);
      
      \draw[->] ([xshift=-5mm]A4.south) 
      .. controls +(down:6mm) and +(up:6mm) ..
      ([xshift=5mm]cA3.north);
      
      \draw[->] (cA3.south) -- (B3.north);
      
      \draw[->] ([xshift=5mm]B2.south)
      .. controls +(down:6mm) and +(up:6mm) ..
      ([xshift=-5mm]cB3.north);
      
      \draw[->] ([xshift=0mm]B3.south) 
      .. controls +(down:6mm) and +(up:6mm) ..
      ([xshift=0]cB3.north);
      
      \draw[->] ([xshift=-5mm]B4.south) 
      .. controls +(down:6mm) and +(up:6mm) ..
      ([xshift=5mm]cB3.north);
      
      \draw[->] (cB3.south) -- (C3.north);
      
      \draw[->] ([xshift=5mm]C2.south)
      .. controls +(down:6mm) and +(up:6mm) ..
      ([xshift=-5mm]cC3.north);
      
      \draw[->] ([xshift=0mm]C3.south) 
      .. controls +(down:6mm) and +(up:6mm) ..
      ([xshift=0]cC3.north);
      
      \draw[->] ([xshift=-5mm]C4.south) 
      .. controls +(down:6mm) and +(up:6mm) ..
      ([xshift=5mm]cC3.north);
      
      \draw[->] (cC3.south) -- (D3.north);
      
      \draw[->] ([xshift=5mm]A3.south)
      .. controls +(down:6mm) and +(up:6mm) ..
      ([xshift=-5mm]cA4.north);
      
      \draw[->] ([xshift=0mm]A4.south) 
      .. controls +(down:6mm) and +(up:6mm) ..
      ([xshift=0]cA4.north);
      
      \draw[->] ([xshift=-5mm]A5.south) 
      .. controls +(down:6mm) and +(up:6mm) ..
      ([xshift=5mm]cA4.north);
      
      \draw[->] (cA4.south) -- (B4.north);
      
      \draw[->] ([xshift=5mm]B3.south)
      .. controls +(down:6mm) and +(up:6mm) ..
      ([xshift=-5mm]cB4.north);
      
      \draw[->] ([xshift=0mm]B4.south) 
      .. controls +(down:6mm) and +(up:6mm) ..
      ([xshift=0]cB4.north);
      
      \draw[->] ([xshift=-5mm]B5.south) 
      .. controls +(down:6mm) and +(up:6mm) ..
      ([xshift=5mm]cB4.north);
      
      \draw[->] (cB4.south) -- (C4.north);
      
      \draw[->] ([xshift=5mm]C3.south)
      .. controls +(down:6mm) and +(up:6mm) ..
      ([xshift=-5mm]cC4.north);
      
      \draw[->] ([xshift=0mm]C4.south) 
      .. controls +(down:6mm) and +(up:6mm) ..
      ([xshift=0]cC4.north);
      
      \draw[->] ([xshift=-5mm]C5.south) 
      .. controls +(down:6mm) and +(up:6mm) ..
      ([xshift=5mm]cC4.north);
      
      \draw[->] (cC4.south) -- (D4.north);
      
    \end{tikzpicture}
  \end{center}
  \caption{The pattern of connections among the registers corresponding to
    individual tape squares, along with copies of the function $f$ that
    performs the local updates on these registers.}
  \label{fig:circuit-pattern}
\end{figure}

Assuming that the states of the $2t+1$ register pairs
$(\reg{S}_{-t},\reg{T}_{-t}),\ldots, (\reg{S}_{t},\reg{T}_{t})$ 
are encoded as binary strings, that the computation represented by the function
$f$ is implemented as a (constant-size) Boolean circuit, and that the
pre-processing and post-processing steps suggested above have size
at most $O(t^2)$ and depth at most $O(t)$, which is ample size and depth to
handle a wide range of input and output encoding schemes, one obtains a Boolean
circuit simulation of the original Turing machine having linear depth and
quadratic size.
The circuits also conform to a simple and regular pattern, and can therefore
be uniformly generated in logarithmic-space (and hence
polynomial-time) by a deterministic Turing machine.

Now, it is not at all clear how this classic deterministic Turing machine
simulation can be extended to a quantum circuit simulation of quantum Turing
machines.
A first guess for how this might be done is to somehow replace the function
$f$ by a unitary operation that describes the evolution of local parts of the
quantum Turing machine---but the function $f$ does not even have input and
output sets of the same cardinality.
One might therefore hope to replace the function $f$ with a unitary operation
that transforms each triple of register pairs
$(\reg{S}_{i-1},\reg{T}_{i-1})$, $(\reg{S}_i,\reg{T}_i)$,
$(\reg{S}_{i+1},\reg{T}_{i+1})$ in a way that is consistent with one step in
the quantum Turing machine's evolution.
Two complications arise: one is that it is unclear how ``overlapping'' unitary
transformations are to be performed in a consistent way, given that these
operations will generally not commute, and another is that some quantum Turing
machine evolutions fail to be unitary when restricted to a finite portion of
the machine's tape.
(For example, even the trivial Turing machine evolution in which the tape head
moves right on each step without changing state or modifying the tape contents
induces a non-invertible transformation when restricted to any finite region of
the tape.)
Yao's simulation method does indeed overcome these obstacles, but requires
additional ideas in order to do this.

\section{Quantum Turing machines}
\label{sec:QTM}

In this section we describe the quantum Turing machine model, which was
first introduced by Deutsch \cite{Deutsch85} and later studied in depth
by Bernstein and Vazirani \cite{BernsteinV93,BernsteinV97}.
Quantum Turing machines generalize Turing machines to the quantum realm
by allowing them to transition in superposition, subject to constraints that
ensure that the overall evolution of the Turing machine is unitary.

Similar conventions will be followed for quantum Turing machines as for
deterministic Turing machines, as discussed in the previous section.
In particular, we assume for simplicity that quantum Turing machines have
a state set of the form $Q = \{1,\ldots,m\}$ and a tape alphabet of the form
$\Gamma = \{0,\ldots,k-1\}$, for some choice of positive integers $m$ and~$k$,
and where the tape symbol 0 represents the blank symbol.
Again the tape is assumed to be two-way infinite, with squares indexed by the
set of integers $\integer$, and computations begin with the state~$1$, with the tape
head scanning the tape square indexed by 0, and with the input string written in
the squares indexed $1,\ldots,n$ and all other tape squares containing
the blank symbol.

The transition function of a quantum Turing machine takes the form
\begin{equation}
  \label{eq:transition-function-form}
  \delta: Q\times\Gamma \rightarrow \complex^{Q\times\Gamma\times\{-1,+1\}},
\end{equation}
meaning that $\delta(p,a)$ is a complex vector indexed by the set
$Q\times\Gamma\times\{-1,+1\}$ for every $p\in Q$ and $a\in\Gamma$.
The interpretation of a transition function $\delta$ of the form
\eqref{eq:transition-function-form} is as follows:
for each choice of states $p,q\in Q$, tape symbols
$a,b\in\Gamma$, and a direction $D\in\{-1,+1\}$, the complex number
$\delta(p,a)[q,b,D]$ represents the \emph{amplitude} with which a quantum
Turing machine whose current state is $p$ and whose tape head is scanning the
symbol $a$ will change state to $q$, write $b$ to the tape, and move
its tape head in direction $D$.
Not all transition functions of this form describe valid quantum Turing
machines; only those transition functions that induce global unitary
evolutions, to be discussed shortly, are valid.
We note that this definition can easily be relaxed so that the tape head of a
quantum Turing machine is not required to move left or right on each step, but
instead can remain stationary, by allowing the transition function to take the
form
\begin{equation}
  \delta: Q\times\Gamma \rightarrow \complex^{Q\times\Gamma\times\{-1,0,+1\}},
\end{equation}
but in the interest of simplicity we will focus on transition functions of the
form \eqref{eq:transition-function-form} in the discussion that follows.

In order to obtain a computational model that does not permit difficult or
impossible to compute information to be somehow hidden inside of a given
transition function, it is important that the complex numbers
$\delta(p,a)[q,b,D]$ are drawn from a reasonable set, such as a finite set
like
\begin{equation}
  \Bigl\{0,\pm 1,\pm i, \pm \frac{1}{\sqrt{2}}, \pm \frac{i}{\sqrt{2}}\Bigr\}
\end{equation}
or a set for which rational approximations can be efficiently computed.
Adleman, DeMarrais, and Huang \cite{AdlemanDH97} discuss the importance of
such assumptions.
For the sake of this paper, however, we will mostly ignore this issue: the
simulation to be analyzed places no restrictions on the complex numbers
appearing in~$\delta$, but any computationally offensive properties possessed
by~$\delta$ will be inherited by the quantum circuits that result from the
simulation.

In order to specify the global evolution of a quantum Turing machine that a
given transition function $\delta$ induces, we must clarify the notion of a
\emph{configuration} of a Turing machine, which is a classical description
of the machine's state, tape head location, and tape contents.
The state and tape head location of a Turing machine correspond to elements of
the sets $Q$ and $\integer$, respectively, while the contents of a Turing
machine tape can be described by a function $T:\integer\rightarrow\Gamma$,
which specifies that the tape symbol $T(j)$ is stored in the tape square
indexed by $j$, for each integer $j$.
The \emph{support} of such a function is defined as
\begin{equation}
  \op{supp}(T) = \{j\in\integer\,:\,T(j) \not= 0\},
\end{equation}
which is the set of tape square indices that do not contain the blank symbol 0.
We are only concerned with those functions $T$ that have finite support, given
that we only consider computations that begin with a finite-length input string
written on an otherwise blank tape, and therefore the set of all configurations
of a Turing machine forms a countably infinite set.

For a given function $T:\integer\rightarrow\Gamma$, an index $i\in\integer$,
and a tape symbol $a\in\Gamma$, let us write $T_{i,a}$ to denote the function
defined as
\begin{equation}
  T_{i,a}(j) = \begin{cases}
    a & \text{if $j = i$}\\
    T(j) & \text{if $j \not= i$}.
  \end{cases}
\end{equation}
Thus, if the contents of a Turing machine tape are described by $T$, and then
the symbol $a$ overwrites the contents of the square indexed by $i$, then
the resulting tape contents are described by $T_{i,a}$.

For a fixed choice of $Q$ and $\Gamma$, let $\H$ denote the Hilbert space of
complex vectors indexed by the set of configurations of a Turing machine with
state set $Q$ and tape alphabet~$\Gamma$.
That is, $\H$ is the Hilbert space whose standard basis includes the vectors
$\ket{p,i,T}$, where $p\in Q$, $i\in\integer$, and
$T:\integer\rightarrow\Gamma$ has finite support.
The global evolution of a quantum Turing machine whose transition function is
$\delta$ can now be specified by the operator $U_{\delta}$ on $\H$ defined by
the action
\begin{equation}
  \label{eq:QTM-evolution}
  U_{\delta} \ket{p,i,T} = \sum_{q, a, D} \delta(p,T(i))[q,a,D]\,
  \ket{q,i+D,T_{i,a}}
\end{equation}
on standard basis states, and extended to all of $\H$ by linearity.

Bernstein and Vazirani \cite{BernsteinV93} identified conditions
on the transition function $\delta$ that cause the operator $U_{\delta}$ to
be unitary.
To be more precise, they identified conditions under which $U_{\delta}$ is an
isometry, and proved that $U_{\delta}$ is necessarily unitary whenever it is an
isometry.
Although the specific conditions they identify are not relevant to this paper,
it is a simple matter to recall them:
\begin{enumerate}
\item[1.]
  The set of vectors $\bigl\{\delta(p,a)\,:\,p\in Q,\; a\in\Gamma\bigr\}$
  is orthonormal.
      
\item[2.]
  For all triples $(p_0,a_0,b_0),(p_1,a_1,b_1)\in Q\times\Gamma\times\Gamma$,
  one has
  \begin{equation}
    \label{eq:QTMcondition2}
    \sum_{q\in Q} \delta(p_0,a_0)[q,b_0,+1]\,
    \overline{\delta(p_1,a_1)[q,b_1,-1]} = 0.
  \end{equation} 
\end{enumerate}
One may note, in particular, that these conditions are easily checked for a
given transition function $\delta$, as they express a finite number of
orthonormality relations.

Classical Turing machine definitions usually specify that some states are to be
considered as \emph{halting states}, with the understanding being that a Turing
machine continues to compute so long as it has not entered a halting state, and
then stops once a halting state is reached.
Stopping conditions for quantum Turing machines are more subtle.
Deutsch \cite{Deutsch85} suggested that periodic measurements could determine
when a quantum Turing machine computation is to be terminated, whereas
Bernstein and Vazirani \cite{BernsteinV93,BernsteinV97} considered quantum
Turing machine computations that run for a predetermined number of steps.
We will adopt the second convention, which is particularly well-suited to the
simulation of quantum Turing machines by quantum circuits: we simply consider
that the number of steps $t$ of a quantum Turing machine to be simulated by a
quantum circuit is fixed and hard-coded into the circuit.

\subsection*{Variants of quantum Turing machines}
\label{subsec:QTM-variants}

As is the case for classical Turing machines, one may consider variants of
quantum Turing machines, such as quantum Turing machines with multiple tapes,
with tapes having a fixed dimension larger than one, with tape heads that have
greater freedom in their movements, and so on.
The quantum Turing machine definition suggested above can be extended to handle
such variants in a natural way.
For example, a quantum Turing machine with 3 tapes could be described by a
transition function of the form
\begin{equation}
  \delta: Q \times \Gamma^3 \rightarrow
  \complex^{Q\times (\Gamma \times \{-1,+1\})^3},
\end{equation}
with $\delta(p,a_1,a_2,a_3)[q,b_1,D_1,b_2,D_2,b_3,D_3]$
indicating the amplitude with which a quantum Turing machine in state $p$ and
reading the symbols $a_1$, $a_2$, and $a_3$ on its tapes will transition to
state $q$, write the symbols $b_1$, $b_2$, and $b_3$ on its tapes, and move the
tape heads in directions $D_1$, $D_2$, and $D_3$.

Unfortunately, it becomes increasingly difficult to check that a given
transition function induces a unitary global evolution when the Turing machine
variant becomes more complex.
For example, Ozawa and Nishamura~\cite{OzawaN00} identified relatively
simple conditions guaranteeing unitary evolutions for quantum Turing machine
transition functions allowing for stationary tape head movements, and rather
complex conditions for two-tape quantum Turing machines.
We do not investigate the difficulty of checking whether a transition function
of a given quantum Turing machine variant induces a unitary evolution,
but simply assume that a transition function must induce a unitary global
evolution in order for it to be considered valid.
If one attempts to apply the simulation method we describe to a quantum Turing
machine whose global evolution is not unitary, it will result in a non-unitary
circuit.

For the most part, we will not focus too much on the technical aspects of any
of the possible variants of quantum Turing machines.
It is the case, however, that the simulation method we describe extends easily
to some interesting variants of the quantum Turing machine model.
This point will be revisited later in the paper after the simulation and its
analysis have been presented.

\subsection*{Quantum Turing machines with looped tapes}
\label{subsec:truncationsAndLoops}

For a given quantum Turing machine $M$ having state set $Q$ and tape alphabet
$\Gamma$, one has that the Hilbert space $\H$ with respect to which the global
quantum states of $M$ are defined is infinite-dimensional.
On the other hand, any finite-length computation of $M$ will only involve a
finite-dimensional subspace of this Hilbert space.
More concretely, if $M$ runs for $t \geq n$ steps on an input of length $n$,
then its tape head will never leave the portion of the tape indexed by elements
of the set $\{-t,\ldots,t\}$, and all tape squares outside of this region will
contain blank symbols for the duration of the computation.

For this reason, it is tempting to imagine the tape has been truncated in such
a case, so that every relevant classical configuration of $M$ takes the form
$(p,i,T)$ for $i\in\{-t,\ldots,t\}$ and
$\op{supp}(T) \subseteq \{-t,\ldots,t\}$. 
One may imagine that any quantum state of $M$ reached at any point during
such a computation is represented by a unit vector in the finite-dimensional
Hilbert space whose standard basis includes precisely those elements
$\ket{p,i,T}$ for which $i\in\{-t,\ldots,t\}$ and
$\op{supp}(T) \subseteq \{-t,\ldots,t\}$. 
A problem arises, however, which is that this space is generally not invariant
under the action of the evolution operator $U_{\delta}$; and defining a
matrix from this operator by discarding rows and columns corresponding to tape
heads or non-blank tape symbols outside of the region indexed by
$\{-t,\ldots,t\}$ may result in a non-unitary (and possibly non-normal)
matrix.

A simple way to address this issue is to imagine that the tape has been
formed into a loop rather than truncated.
Specifically, for every positive integer $N$, one may consider a Turing machine
tape loop whose squares are indexed by the set
$\integer_N = \{0,1,\ldots,N-1\}$, and where tape head movements are calculated
modulo $N$.
Specifically, we will define a finite-dimensional Hilbert space $\H_N$ whose
standard basis contains all elements of the form $\ket{p,i,T}$ where $p\in Q$,
$i\in\integer_N$, and $T$ takes the form
\begin{equation}
  T:\integer_N \rightarrow\Gamma.
\end{equation}
(No assumption on the finiteness of the support of $T$ is required in this
case, of course, as $\integer_N$ is finite.)
The transition function $\delta$ now defines an operator
\begin{equation}
  U_{\delta,N} \in \Lin(\H_N)
\end{equation}
for every choice of a positive integer $N$ by precisely the same formula 
\eqref{eq:QTM-evolution} as before, except that the expression $i+D$ is
understood to refer to addition modulo $N$.
We observe that if $U_{\delta}$ is unitary, then $U_{\delta, N}$ is necessarily
unitary for every choice of $N \geq 5$.

A simulation of a quantum Turing machine on an input string of length $n$ for
$t\geq n$ steps can immediately be obtained from a simulation of the same
machine on a tape loop of size $N = 2t + 1$ (or any choice of $N$ larger
than $2t+1$).
For this reason we will focus on quantum Turing machines with tapes formed
into loops, whose size will be a function of the input length and number of
steps for which the machine is to be simulated.
Viewing quantum Turing machine computations as taking place on a tape loop is
just a minor convenience that allows us to work entirely with
finite-dimensional Hilbert spaces and removes the need for special cases at the
edges of the tape region indexed by $-t,\ldots,t$.

\section{Localizing causal unitary evolutions}
\label{sec:localizability}

This section of the paper is concerned with a relationship between causality
and localizability of unitary operators, first described by Arrighi, Nesme, and
Werner \cite{ArrighiNW11}.
Our presentation of this relationship will differ somewhat from theirs,
however, and we will require a generalization of their findings that is
concerned with operators whose causality holds only on certain subspaces of a
tensor product space.
Although the relationship between causality and localizability to be discussed
is a key to our analysis of the simulation of quantum Turing machines by
quantum circuits, the section itself is independent of quantum Turing
machines.

We will begin with a definition of causality, which includes both the cases in
which causality holds on an entire space (as considered by
Arrighi, Nesme, and Werner) or just on a subspace.

\begin{definition}
  \label{def:causality}
  Let $\reg{X}$, $\reg{Y}$, and $\reg{Z}$ be registers having associated
  Hilbert spaces $\X$, $\Y$, and $\Z$, respectively, and let
  $U\in\Unitary(\X\otimes\Y\otimes\Z)$ be a unitary operator.
  \begin{enumerate}
  \item[1.]
    The operator $U$ is $\reg{Y}\rightarrow\reg{X}$ causal if, for every pair
    of states $\rho,\sigma\in\Density(\X\otimes\Y\otimes\Z)$ satisfying
    $\tr_{\Z}(\rho) = \tr_{\Z}(\sigma)$, one has
    \begin{equation}
      \tr_{\Y\otimes\Z}(U \rho U^{\ast})
      = \tr_{\Y\otimes\Z}(U \sigma U^{\ast}).
    \end{equation}
  \item[2.]
    The operator $U$ is $\reg{Y}\rightarrow\reg{X}$ causal on a subspace
    $\V\subseteq\X\otimes\Y\otimes\Z$ if, for every pair of states
    $\rho,\sigma\in\Density(\X\otimes\Y\otimes\Z)$ satisfying
    $\im(\rho)\subseteq\V$, $\im(\sigma)\subseteq\V$, and
    $\tr_{\Z}(\rho) = \tr_{\Z}(\sigma)$, one has
    \begin{equation}
      \tr_{\Y\otimes\Z}(U \rho U^{\ast})
      = \tr_{\Y\otimes\Z}(U \sigma U^{\ast}).
    \end{equation}
  \end{enumerate}
\end{definition}

The intuition behind this definition is that a unitary operator $U$ is
$\reg{Y}\rightarrow\reg{X}$ causal if the state of $\reg{X}$ after the
application of the operator $U$ is completely determined by the state of
$(\reg{X},\reg{Y})$ before the application of $U$.
A natural way to view this situation is that $\reg{X}$ represents
some local region of interest, $\reg{Y}$ represents a neighborhood of $\reg{X}$
(excluding $\reg{X}$ itself), and $\reg{Z}$ represents everything outside of
this neighborhood.
If $U$ is $\reg{Y}\rightarrow\reg{X}$ causal, then changes to $\reg{X}$ induced
by $U$ are effectively caused by the state of $(\reg{X},\reg{Y})$ and are
not influenced by the state of $\reg{Z}$.
The restriction of this property to a subspace $\V$ of $\X\otimes\Y\otimes\Z$
requires only that this property holds for states fully supported on $\V$.

Next we will prove two lemmas that lead naturally to the main result of the
section.
The first lemma establishes a simple but useful technical condition on causal
unitary operators.

\begin{lemma}
  \label{lemma:localizability-1}
  Let $\reg{X}$, $\reg{Y}$, and $\reg{Z}$ be registers having associated
  Hilbert spaces $\X$, $\Y$, and $\Z$, respectively, let
  $\V\subseteq\X\otimes\Y\otimes\Z$ be a subspace, and let
  $U\in\Unitary(\X\otimes\Y\otimes\Z)$ be a $\reg{Y}\rightarrow\reg{X}$ causal
  unitary operator on the subspace $\V$.
  For every Hermitian operator $H\in\Herm(\X\otimes\Y\otimes\Z)$ satisfying
  $\im(H)\subseteq\V$ and $\tr_{\Z}(H) = 0$, and every operator
  $X\in\Lin(\X)$, one has
  \begin{equation}
    \ip{H}{U^{\ast} (X\otimes\I_{\Y}\otimes\I_{\Z})U} = 0.
  \end{equation}
\end{lemma}

\begin{proof}
  The statement is trivial when $H=0$, so assume $H$ is nonzero, and let
  \begin{equation}
    \label{eq:Jordan-Hahn-of-H}
    H = P - Q
  \end{equation}
  be the Jordan-Hahn decomposition of $H$, meaning that
  $P,Q\in\Pos(\X\otimes\Y\otimes\Z)$ are the unique positive semidefinite
  operators satisfying \eqref{eq:Jordan-Hahn-of-H} and $PQ = 0$.
  The assumption $\tr_{\Z}(H) = 0$ implies that $\tr_{\Z}(P) = \tr_{\Z}(Q)$,
  and in particular
  \begin{equation}
    \tr(P) = \lambda = \tr(Q)
  \end{equation}
  for some positive real number $\lambda > 0$.
  Moreover, by the assumption $\im(H)\subseteq\V$, one has that
  $\im(P)\subseteq\V$ and $\im(Q)\subseteq\V$.
  Thus, the density operators $\rho = P/\tr(P)$ and $\sigma = Q/\tr(Q)$ satisfy
  \begin{equation}
    H = \lambda (\rho - \sigma) \quad\text{and}\quad
    \tr_{\Z}(\rho) = \tr_{\Z}(\sigma),
  \end{equation}
  as well as $\im(\rho)\subseteq\V$ and $\im(\sigma)\subseteq\V$.
  By the assumption that $U$ is causal on $\V$, it follows that
  \begin{equation}
    \ip{H}{U^{\ast} (X\otimes\I_{\Y}\otimes\I_{\Z})U}
    = \lambda \ip{\tr_{\Y\otimes\Z}(U\rho U^{\ast}) -
      \tr_{\Y\otimes\Z}(U\sigma U^{\ast})}{X} = 0,
  \end{equation}
  as required.
\end{proof}

The second lemma draws an implication from the structure suggested in the
previous lemma that will lead naturally to the notion of localizability,
provided that the subspace in question is suitably aligned with the underlying
tensor product structure of the global space.
(Note that the Hilbert space $\W$ in this lemma plays the role of $\X\otimes\Y$
in the definition of causality.)

\begin{lemma}
  \label{lemma:localizability-2}
  Let $\W$ and $\Z$ be finite-dimensional Hilbert spaces, let
  $\{\Delta_1,\ldots,\Delta_n\}\subseteq\Proj(\W)$ and
  $\{\Lambda_1,\ldots,\Lambda_n\}\subseteq\Proj(\Z)$
  be orthogonal sets of nonzero projection operators, and let
  \begin{equation}
    \label{eq:special-form-of-Pi}
    \Pi = \sum_{k = 1}^n \Delta_k \otimes \Lambda_k.
  \end{equation}
  For every operator $A \in \Lin(\W\otimes\Z)$ such that $\ip{H}{A}=0$ for all
  Hermitian operators $H$ with $\im(H) \subseteq \im(\Pi)$ and $\tr_{\Z}(H)=0$,
  there exists an operator $W\in\Lin(\W)$ such that
  \begin{equation}
    \label{eq:localizable-on-V}
    \Pi A\Pi = \Pi(W\otimes\I_{\Z})\Pi.
  \end{equation}
  If, in addition, $A$ is a unitary operator and $[A,\Pi] = 0$, then there
  exists a unitary operator $W$ that satisfies \eqref{eq:localizable-on-V}.
\end{lemma}

\begin{proof}
  Suppose $Z\in\Unitary(\Z)$ is any unitary operator satisfying
  \begin{equation}
    \label{eq:commutation-Z-and-Pi}
    [\I_{\X}\otimes Z,\Pi] = 0.
  \end{equation}
  For an arbitrarily chosen Hermitian operator
  $H\in\Herm(\X\otimes\Z)$, one has
  \begin{equation}
    \label{eq:zero-inner-product}
    \ip{H}{\Pi A\Pi - (\I_{\X}\otimes Z)\Pi A \Pi(\I_{\X}\otimes Z^{\ast})}
    = \ip{K}{A}
  \end{equation}
  for
  \begin{equation}
    K = \Pi H \Pi -
    (\I_{\X}\otimes Z^{\ast}) \Pi H \Pi (\I_{\X}\otimes Z).
  \end{equation}
  The operator $K$ is Hermitian and satisfies $\im(K) \subseteq \im(\Pi)$ and
  $\tr_{\Z}(K) = 0$, and therefore by the assumptions of the lemma the quantity
  represented by \eqref{eq:zero-inner-product} is zero.
  Because this is so for every choice of $H\in\Herm(\X\otimes\Z)$, it follows
  that
  \begin{equation}
    \Pi A \Pi = (\I_{\X}\otimes Z) \Pi A \Pi (\I_{\X}\otimes Z^{\ast}),
  \end{equation}
  which is equivalent to
  \begin{equation}
    \label{eq:commutation-relation}
    \bigl[ \I_{\X}\otimes Z, \Pi A \Pi \bigr] = 0.
  \end{equation}

  Now, the set of all unitary operators $Z\in\Unitary(\Z)$ for which
  \eqref{eq:commutation-Z-and-Pi} is satisfied includes those operators for
  which $[Z,\Lambda_k] = 0$ for all $k\in\{1,\ldots,n\}$, from which it follows
  that
  \begin{equation}
    \label{eq:sum-expression-causal}
    \Pi A \Pi = \sum_{k = 1}^n W_k \otimes \Lambda_k
  \end{equation}
  for some choice of $W_1,\ldots,W_n\in\Lin(\X)$ satisfying
  $W_k = \Delta_k W_k \Delta_k$ for each $k\in\{1,\ldots,n\}$.
  By setting
  \begin{equation}
    \label{eq:definition-of-W}
    W = W_1 + \cdots + W_n + \bigl(\I_{\X}-\Delta_1-\cdots-\Delta_n\bigr),
  \end{equation}
  one obtains an operator satisfying
  \begin{equation}
    \label{eq:localityRelation}
    \Pi A \Pi = \sum_{k = 1}^n \Delta_k W \Delta_k \otimes \Lambda_k
    = \Pi(W\otimes\I_{\Z})\Pi.
  \end{equation}
  
  Finally, if $A$ is unitary and $[A,\Pi]=0$, then $A$ is a unitary operator
  when restricted to $\im(\Pi)$, which implies that each $W_k$ in
  \eqref{eq:sum-expression-causal} is unitary when restricted to
  $\im(\Delta_k)$.
  The operator $W$ defined in \eqref{eq:definition-of-W} is therefore unitary,
  which completes the proof.
\end{proof}

\begin{remark}
  Note that if one allows $\{\Delta_1,\ldots,\Delta_n\}$ not to be mutually
  orthogonal, the decomposition for $\Pi A \Pi$ given by
  \eqref{eq:sum-expression-causal} can still be obtained.
  However, now the operator defined in \eqref{eq:definition-of-W} does not
  satisfy the relation in \eqref{eq:localityRelation} because the images of
  the operators $W_1,\ldots,W_n$ may overlap.
\end{remark}

\begin{remark}
Note that our proof for Lemma \ref{lemma:localizability-2} will still go through if the space $\Z$ is allowed to be an infinite-dimensional separable Hilbert space, as can be checked for example through $\cite{conway2000course, takesaki2013theory}$.
\end{remark}

Finally, we state the main theorem of the section, which connects the property
of a unitary operator being causal with the notion of \emph{localizability},
which simply means that an operator can be represented as a tensor product of
one operator with the identity operator.
The original result of Arrighi, Nesme, and Werner that this theorem generalizes
states that if $U$ is a $\reg{Y}\rightarrow\reg{X}$ causal unitary operator and
$X\in\Unitary(\X)$ is a unitary operator on $\X$, then
\begin{equation}
  U^{\ast}(X\otimes\I_{\Y\otimes\Z})U = W\otimes\I_{\Z}
\end{equation}
for some unitary operator $W\in\Unitary(\X\otimes\Y)$.
As has already been suggested, the generalization represented by the theorem
that follows concerns unitary operators that are only causal on some subspace
that is suitably aligned with the tensor product structure of
$\X\otimes\Y\otimes\Z$.

\begin{theorem}
  \label{theorem:localizability}
  Let $\X$, $\Y$, and $\Z$ be complex Euclidean spaces, let
  $\{\Delta_1,\ldots,\Delta_n\}\subseteq\Proj(\X\otimes\Y)$ and
  $\{\Lambda_1,\ldots,\Lambda_n\}\subseteq\Proj(\Z)$ be orthogonal sets of
  nonzero projection operators, let
  \begin{equation}
    \label{eq:special-form-of-Pi2}
    \Pi = \sum_{k = 1}^n \Delta_k \otimes \Lambda_k,
  \end{equation}
  and let $U\in\Unitary(\X\otimes\Y\otimes\Z)$ be a unitary operator that is
  $\reg{Y}\rightarrow\reg{X}$ causal on $\im(\Pi)$.
  For every operator $X\in\Lin(\X)$, there exists an operator
  $W\in\Lin(\X\otimes\Y)$ such that
  \begin{equation}
    \label{eq:localizable-on-V2}
    \Pi U^{\ast}(X\otimes\I_{\Y\otimes\Z})U\Pi = \Pi(W\otimes\I_{\Z})\Pi.
  \end{equation}
  If, in addition, $X$ is unitary and
  $[U^{\ast}(X\otimes\I_{\Y\otimes\Z})U,\Pi] = 0$, then $W$ may also be taken
  to be unitary.
\end{theorem}

\begin{proof}
  Let $\W = \X\otimes\Y$ and $A = U^{\ast}(X\otimes\I_{\Y\otimes\Z})U$.
  By Lemma~\ref{lemma:localizability-1} the operator $A$ satisfies the
  requirements of Lemma~\ref{lemma:localizability-2}, which in turn implies the
  theorem.
\end{proof}

\section{Circuit simulation of quantum Turing machines}
\label{sec:simulation}

In this section we present a simulation of quantum Turing machines by quantum
circuits based on the simulation method of Yao, together with its analysis.
As has already been mentioned, there are technical differences between the
simulation we present and the one originally proposed by Yao, and these
differences are described briefly later in the section.

\subsection*{Simulation structure}

Following the same conventions that were described in Section~\ref{sec:QTM},
we will assume that the quantum Turing machine $M$ to be simulated has state
set $Q = \{1,\ldots,m\}$ and tape alphabet $\Gamma = \{0,\ldots,k-1\}$, and has
a transition function
\begin{equation}
  \delta:Q\times\Gamma\rightarrow\complex^{Q\times\Gamma\times\{-1,+1\}}
\end{equation}
that induces a unitary global evolution.
It will also be assumed that some input alphabet
$\Sigma\subseteq\{1,\ldots,k-1\}$ has been specified, and that the computation
of the quantum Turing machine $M$ on an input string $x\in\Sigma^n$ is to be
simulated for $t\geq n$ steps.
The simulation can be performed for fewer than $n$ steps, but the assumption
that $t\geq n$ allows us to write $t$ rather than $\max\{n,t\}$ in various
places throughout the proof, and little generality is lost in disregarding
Turing machine computations that are not even long enough to read their entire
input string.

Quantum circuits operate on qubits, of course, but it is convenient to first
describe a circuit simulation of quantum Turing machines that operates on
collections of registers whose classical state sets relate to the sets $\Gamma$
and $Q$, rather than on qubits.
(The classical Turing machine simulation in Section~\ref{sec:DTM} was described
in a similar style.)
More precisely, the simulation described below makes use of registers whose
classical state sets are either $\{0,\ldots,k-1\}$ or $\{-m,\ldots,m\}$.
The gates in these circuits will operate on at most six registers, three of
each of the two sizes just mentioned.
In both cases, these are constant-size registers, each such register can
be replaced by a constant number of qubits, and the operations on these
registers that appear in the simulation can be replaced by constant-size quantum circuits
acting on these qubits. In such a replacement, elements of the sets $\{0,\ldots,k-1\}$ or
$\{-m,\ldots,m\}$ would be encoded as binary strings of the appropriate length, and
an arbitrary choice for such encodings may be selected.

The simulation includes pre-processing and post-processing steps that will be
discussed shortly.
The main part of the simulation functions in an iterative manner that resembles
the classical simulation described in Section~\ref{sec:DTM}.
That is, it consists of a concatenation of $t$ identical circuit \emph{layers},
each of which simulates a single step of the Turing machine.
Also similar to the classical case, the simulation will make use of a
collection of registers
\begin{equation}
  \label{eq:register-pairs}
  (\reg{S}_{-t},\reg{T}_{-t}), \ldots, (\reg{S}_{t},\reg{T}_{t})
\end{equation}
to represent those tape squares indexed by integers in the range
$\{-t,\ldots,t\}$; each register $\reg{S}_i$ indicates the presence or absence
of the tape head at the square indexed by $i$, as well as the Turing machine's
state if the head is present at this location, while $\reg{T}_i$ represents the
contents of the tape square indexed by $i$.

We will set $N = 2t+1$, which is assumed to be at least 5, and imagine that the
Turing machine $M$ runs on a tape loop of length $N$ rather than a two-way
infinite tape.
(One may choose $N$ to be larger than $2t+1$ without compromising the
simulation.
It turns out that it is both natural and convenient to choose $N$ to be the
smallest multiple of 3 that is at least $2t+1$, as will become clear later in
the section.)
As was mentioned in Section~\ref{sec:QTM}, there is essentially no difference
between the two cases, as the tape head never has time to cross the division
between the tape squares indexed by $t$ and $-t$ (or, equivalently,
$t$ and $t+1$, as tape square indices are equated modulo $N$).
The transition function $\delta$ induces a unitary operator $U_{\delta,N}$ on
the Hilbert space $\H_N$ whose standard basis corresponds to the set of
possible configurations of $M$ running on a tape loop of length $N$, as defined
in Section~\ref{sec:QTM}.
Hereafter we will write $U$ rather than $U_{\delta,N}$ for brevity, as $\delta$
and $N$ may safely be viewed as being fixed for the purposes of this
description.

Next, let us be more precise about the registers \eqref{eq:register-pairs}.
As suggested above, we will view the indices of these registers as representing
elements of $\integer_N$, so that they may alternatively be written (without
changing their order) as
$(\reg{S}_{N-t},\reg{T}_{N-t}), \ldots, (\reg{S}_{N-1},\reg{T}_{N-1})$,
$(\reg{S}_{0},\reg{T}_{0}),\ldots, (\reg{S}_{t},\reg{T}_{t})$.
The classical state set of each register $\reg{S}_i$ is $\{-m,\ldots,m\}$ and
the classical state set of each $\reg{T}_i$ is $\{0,\ldots,k-1\}$.
Note, in particular, that this choice differs from the classical case, in which
each $\reg{S}_i$ stores an element of $\{0,\ldots,m\}$.
In essence, the states $-1,\ldots,-m$ will function as ``inactive copies'' of
the states $1,\ldots,m$; this is a simple but key trick that allows Yao's
simulation method to work.
For each $i\in\integer_N$ we will let $\S_i$ and $\T_i$ denote the Hilbert
spaces associated with $\reg{S}_i$ and $\reg{T}_i$, respectively, and we will
let
\begin{equation}
  \K_N = (\S_0\otimes\T_0) \otimes \cdots \otimes (\S_{N-1}\otimes\T_{N-1})
\end{equation}
be the combined Hilbert space of the entire sequence
$(\reg{S}_0,\reg{T}_0),\, \ldots, \, (\reg{S}_{N-1},\reg{T}_{N-1})$.
The spaces $\S_0,\ldots,\S_{N-1}$ are of course equivalent to one another, as
are the spaces $\T_0,\ldots,\T_{N-1}$, and when we wish to refer generally to
any one of these spaces without specifying which one it is, we will simply
write $\S$ or $\T$ without a subscript.

For each configuration $(p,i,T)$ of $M$ running on a tape loop of length $N$,
one may associate a classical state
\begin{equation}
  f(p,i,T) \in \bigl( \{0,\ldots,m\}\times\{0,\ldots,k-1\} \bigr)^N
\end{equation}
of the register pairs
$(\reg{S}_0,\reg{T}_0),\, \ldots, \, (\reg{S}_{N-1},\reg{T}_{N-1})$ in a
similar way to what is done in the classical simulation described in
Section~\ref{sec:DTM}.
That is, each register $\reg{T}_j$ stores $T(j)$ for $j \in \integer_N$, the
register $\reg{S}_i$ stores $p$, which is an element of the set
$\{1,\ldots,m\}$, and every other register $\reg{S}_j$, for $j\not=i$, stores
0.
None of the registers $\reg{S}_0,\ldots,\reg{S}_{N-1}$ stores a negative
value.
One can also define an isometry $A\in\Unitary(\H_N,\K_N)$ based on this correspondence
between configurations and classical register states as
\begin{equation}
  A = \sum_{(p,i,T)} \ket{f(p,i,T)} \bra{p,i,T},
\end{equation}
where the sum is over all configurations of $M$ on a tape loop of length $N$.
It may be observed that the projection $A A^{\ast}$ is alternatively
described as the projection onto the space spanned by classical states of the
registers $(\reg{S}_0,\reg{T}_0),\, \ldots, \, (\reg{S}_{N-1},\reg{T}_{N-1})$
that correspond to valid Turing machine configurations, meaning that none of
the registers $\reg{S}_0,\ldots,\reg{S}_{N-1}$ contain negative values and
exactly one of these registers contains a positive value.

\subsection*{Simulation procedure}

The pre-processing step of the simulation initializes
$(\reg{S}_0,\reg{T}_0)$, \ldots, $(\reg{S}_{N-1},\reg{T}_{N-1})$ to the
standard basis state corresponding to the initial configuration of $M$ on input
$x$.
Each circuit layer in the main part of the simulation will induce a
unitary transformation that agrees with $A U A^{\ast}$ on the subspace $\im(A)$.
The final post-processing step transforms the state of the registers
$(\reg{S}_0,\reg{T}_0)$, \ldots, $(\reg{S}_{N-1},\reg{T}_{N-1})$ into whatever
output form is desired for the simulation.
The state of these registers will be fully supported on $\im(A)$, so a
standard basis measurement of these registers after the completion of the main
part of the simulation would necessarily yield a state corresponding to a valid
configuration of $M$.

The main challenge of the simulation is to efficiently implement the layers in
the main part of the simulation, each of which performs a transformation that
agrees with $A U A^{\ast}$ on $\im(A)$.
To this end, define a reversible (and therefore unitary) transformation
\begin{equation}
  F \ket{p} = \ket{-p}
\end{equation}
for every $p\in\{-m,\ldots,m\}$, which may be regarded as an operator acting on
$\S$, or equivalently on $\S_i$ for any choice of $i\in\integer_N$.
Note that $F$ is not a phase flip, it is a permutation of the standard basis
states $\ket{-m},\ldots,\ket{m}$.
For each index $i\in\integer_N$, define $F_i$ to be the unitary
operator acting on $\K_N$ that is obtained by tensoring $F$ on the register
$\reg{S}_i$ with the identity operator on all of the remaining registers.
Notice that the operators $F_0,\ldots,F_{N-1}$ mutually commute, and that their
product $F_0\cdots F_{N-1}$ is equivalent to $F$ being performed independently
on every one of the registers $\reg{S}_0,\ldots,\reg{S}_{N-1}$.

Next, define two unitary operators:
\begin{equation}
  V = A U A^{\ast} + (\I - A A^{\ast})
  \quad\text{and}\quad
  W = (F_0\cdots F_{N-1}) V^{\ast} (F_0\cdots F_{N-1}) V.
\end{equation}
The subspace $\im(A)$ is evidently an invariant subspace of $V$, as is its
orthogonal complement $\im(A)^{\perp}$, upon which $V$ acts trivially.
The operator $F_0\cdots F_{N-1}$ maps $\im(A)$ into $\im(A)^{\perp}$, as each
standard basis state corresponding to a configuration is transformed into a
standard basis state for which one of the registers
$\reg{S}_0,\ldots,\reg{S}_{N-1}$ contains a negative value.
Specifically, $\ket{f(p,i,T)}$ is transformed into a standard basis state in
which the register $\reg{S}_i$ contains the value $-p$.
The operator
\begin{equation}
  (F_0\cdots F_{N-1}) V^{\ast} (F_0\cdots F_{N-1})
\end{equation}
therefore acts trivially on $\im(A)$, implying that $V$ and $W$ act identically
on $\im(A)$.

Now consider the operator $V^{\ast}(F_0\cdots F_{N-1})V$, which may
alternatively be written
\begin{equation}
  V^{\ast}(F_0\cdots F_{N-1})V =
  (V^{\ast}F_0 V) \cdots (V^{\ast}F_{N-1}V).
\end{equation}
As $F_0,\ldots,F_{N-1}$ mutually commute, so do the operators
$V^{\ast}F_0 V, \ldots, V^{\ast}F_{N-1}V$.
It therefore suffices that each circuit layer in the main part of the
simulation applies these operators, in an arbitrary order, followed by the
operator $F_0\cdots F_{N-1}$ (or, equivalently, $F$ applied independently to
each of the registers $\reg{S}_0,\ldots,\reg{S}_{N-1}$).

\subsection*{Locality}

It remains to prove that each of the operators $V^{\ast} F_i V$ can be
localized, specifically to the registers $(\reg{S}_{i-1},\reg{T}_{i-1}),
(\reg{S}_i,\reg{T}_i), (\reg{S}_{i+1},\reg{T}_{i+1})$, when restricted to
a suitable subspace that contains $\im(A)$.
Fix $i\in\integer_N$, define
$\reg{X} = (\reg{S}_i,\reg{T}_i)$ and 
$\reg{Y} = (\reg{S}_{i-1},\reg{T}_{i-1},\reg{S}_{i+1},\reg{T}_{i+1})$,
and let $\reg{Z}$ denote all of the remaining registers among
$(\reg{S}_0,\reg{T}_0)$, \ldots, $(\reg{S}_{N-1},\reg{T}_{N-1})$ that do not
appear in $\reg{X}$ or $\reg{Y}$.
For each $a\in\{0,1\}$, define $\Delta_a\in\Proj(\X\otimes\Y)$ to be the
projection onto the space spanned by standard basis states of
$(\reg{X},\reg{Y})$ in which precisely $a$ of the registers
$\reg{S}_{i-1},\reg{S}_{i},\reg{S}_{i+1}$ contain a nonzero (either positive
or negative) value, and define $\Lambda_a\in\Proj(\Z)$ similarly, but replacing
$\reg{S}_{i-1},\reg{S}_{i},\reg{S}_{i+1}$ with those registers among
$\reg{S}_0,\ldots,\reg{S}_{N-1}$ that appear in $\reg{Z}$ rather than
$(\reg{X},\reg{Y})$.
Finally, define a projection
\begin{equation}
  \label{eq:one-head-projection}
  \Pi = \Delta_0 \otimes \Lambda_1 + \Delta_1\otimes\Lambda_0.
\end{equation}
In words, $\Pi$ is the projection onto the space spanned by standard basis
states of the registers
$(\reg{S}_0,\reg{T}_0),\ldots,(\reg{S}_{N-1},\reg{T}_{N-1})$ in which exactly
one of the registers $\reg{S}_0,\ldots,\reg{S}_{N-1}$ contains a nonzero value
(representing exactly one tape head, either active or inactive).
The expression \eqref{eq:one-head-projection} reveals that this projection
is aligned with the tensor product structure of $\X\otimes\Y\otimes\Z$ in
a suitable way to allow for Theorem~\ref{theorem:localizability} to be applied
to the situation under consideration.

\begin{figure}[!t]
  \begin{center}
    \begin{tikzpicture}[>=latex]
      \tikzstyle{block}=[draw, minimum width=25mm, minimum height=10mm]
      \tikzstyle{circuit}=[draw, minimum width=85mm, minimum height=10mm,
        fill=black!12]
      \tikzstyle{smallcircuit}=[draw, minimum width=25mm, minimum height=10mm,
        fill=black!12]
      \tikzstyle{register}=[draw, minimum width=10mm, minimum height=4mm,
        font=\fontsize{10}{0}\selectfont]
      
      \node[block] (A1) at (-4.5,8) {};
      \node[block] (A2) at (-1.5,8) {};
      \node[block] (A3) at (1.5,8) {};
      \node[block] (A4) at (4.5,8) {};
      \node[block] (A5) at (7.5,8) {};
      
      \node[block] (B1) at (-4.5,0.5) {};
      \node[block] (B2) at (-1.5,0.5) {};
      \node[block] (B3) at (1.5,0.5) {};
      \node[block] (B4) at (4.5,0.5) {};
      \node[block] (B5) at (7.5,0.5) {};
      
      \node[circuit,yshift=-15mm] (c1) at (A2) {$G$};
      \node[circuit,yshift=-30mm] (c2) at (A3) {$G$};
      \node[circuit,yshift=-45mm] (c3) at (A4) {$G$};

      \node[smallcircuit,yshift=-60mm] (f1) at (A1) {$F\otimes\I$};
      \node[smallcircuit,yshift=-60mm] (f2) at (A2) {$F\otimes\I$};
      \node[smallcircuit,yshift=-60mm] (f3) at (A3) {$F\otimes\I$};
      \node[smallcircuit,yshift=-60mm] (f4) at (A4) {$F\otimes\I$};
      \node[smallcircuit,yshift=-60mm] (f5) at (A5) {$F\otimes\I$};

      \draw[->] (f1) -- (B1);
      \draw[->] (f2) -- (B2);
      \draw[->] (f3) -- (B3);
      \draw[->] (f4) -- (B4);
      \draw[->] (f5) -- (B5);
      
      \draw[fill=black!12]
      (3.25,7) -- (9,7) 
      decorate [decoration={random steps, segment length=.25cm}]%
      {-- (9,6)} -- (3.25,6) -- (3.25,7);
      
      \draw[fill=black!12]
      (6.25,5.5) -- (9,5.5) 
      decorate [decoration={random steps, segment length=.25cm}]%
      {-- (9,4.5)} -- (6.25,4.5) -- (6.25,5.5);
      
      \draw[fill=black!12]
      (-3.25,5.5) -- (-6,5.5) 
      decorate [decoration={random steps, segment length=.25cm}]%
      {-- (-6,4.5)} -- (-3.25,4.5) -- (-3.25,5.5);
      
      \draw[fill=black!12]
      (-.25,4) -- (-6,4) 
      decorate [decoration={random steps, segment length=.25cm}]%
      {-- (-6,3)} -- (-.25,3) -- (-.25,4);
      
      \node[yshift=-15mm] at (A5) {$G$};
      \node[yshift=-45mm] at (A1) {$G$};
        
      \node[register,xshift=-5mm] at (A1) {$\reg{S}_{i-2}$};
      \node[register,xshift=5mm] at (A1) {$\reg{T}_{i-2}$};
      
      \node[register,xshift=-5mm] at (A2) {$\reg{S}_{i-1}$};
      \node[register,xshift=5mm] at (A2) {$\reg{T}_{i-1}$};
      
      \node[register,xshift=-5mm] at (A3) {$\reg{S}_{i}$};
      \node[register,xshift=5mm] at (A3) {$\reg{T}_{i}$};
      
      \node[register,xshift=-5mm] at (A4) {$\reg{S}_{i+1}$};
      \node[register,xshift=5mm] at (A4) {$\reg{T}_{i+1}$};
      
      \node[register,xshift=-5mm] at (A5) {$\reg{S}_{i+2}$};
      \node[register,xshift=5mm] at (A5) {$\reg{T}_{i+2}$};
      
      \node[register,xshift=-5mm] at (B1) {$\reg{S}_{i-2}$};
      \node[register,xshift=5mm] at (B1) {$\reg{T}_{i-2}$};
      
      \node[register,xshift=-5mm] at (B2) {$\reg{S}_{i-1}$};
      \node[register,xshift=5mm] at (B2) {$\reg{T}_{i-1}$};
      
      \node[register,xshift=-5mm] at (B3) {$\reg{S}_{i}$};
      \node[register,xshift=5mm] at (B3) {$\reg{T}_{i}$};
      
      \node[register,xshift=-5mm] at (B4) {$\reg{S}_{i+1}$};
      \node[register,xshift=5mm] at (B4) {$\reg{T}_{i+1}$};
      
      \node[register,xshift=-5mm] at (B5) {$\reg{S}_{i+2}$};
      \node[register,xshift=5mm] at (B5) {$\reg{T}_{i+2}$};

      \draw[->] ([yshift=8mm]A1.north) -- (A1.north);
      \draw[->] ([yshift=8mm]A2.north) -- (A2.north);
      \draw[->] ([yshift=8mm]A3.north) -- (A3.north);
      \draw[->] ([yshift=8mm]A4.north) -- (A4.north);
      \draw[->] ([yshift=8mm]A5.north) -- (A5.north);

      \draw[->] ([yshift=-10mm]A1.north) -- ([yshift=-15mm]A1.north);
      \draw[->] ([yshift=-25mm]A1.north) -- ([yshift=-30mm]A1.north);
      \draw[->] ([yshift=-40mm]A1.north) -- ([yshift=-45mm]A1.north);
      \draw[->] ([yshift=-55mm]A1.north) -- ([yshift=-60mm]A1.north);

      \draw[->] ([yshift=-10mm]A2.north) -- ([yshift=-15mm]A2.north);
      \draw[->] ([yshift=-25mm]A2.north) -- ([yshift=-30mm]A2.north);
      \draw[->] ([yshift=-40mm]A2.north) -- ([yshift=-45mm]A2.north);
      \draw[->] ([yshift=-55mm]A2.north) -- ([yshift=-60mm]A2.north);

      \draw[->] ([yshift=-10mm]A3.north) -- ([yshift=-15mm]A3.north);
      \draw[->] ([yshift=-25mm]A3.north) -- ([yshift=-30mm]A3.north);
      \draw[->] ([yshift=-40mm]A3.north) -- ([yshift=-45mm]A3.north);
      \draw[->] ([yshift=-55mm]A3.north) -- ([yshift=-60mm]A3.north);

      \draw[->] ([yshift=-10mm]A4.north) -- ([yshift=-15mm]A4.north);
      \draw[->] ([yshift=-25mm]A4.north) -- ([yshift=-30mm]A4.north);
      \draw[->] ([yshift=-40mm]A4.north) -- ([yshift=-45mm]A4.north);
      \draw[->] ([yshift=-55mm]A4.north) -- ([yshift=-60mm]A4.north);

      \draw[->] ([yshift=-10mm]A5.north) -- ([yshift=-15mm]A5.north);
      \draw[->] ([yshift=-25mm]A5.north) -- ([yshift=-30mm]A5.north);
      \draw[->] ([yshift=-40mm]A5.north) -- ([yshift=-45mm]A5.north);
      \draw[->] ([yshift=-55mm]A5.north) -- ([yshift=-60mm]A5.north);
      
    \end{tikzpicture}
  \end{center}
  \caption{
    Each Turing machine step is simulated by one circuit layer in the main part
    of the simulation, which consists of a pattern of unitary operations as
    illustrated.
    (Note that if the number of register pairs is not divisible by 3, then one
    or two additional $G$ transformations may need to be performed on separate
    levels so that $G$ is applied once to each
    triple $(\reg{S}_{i-1},\reg{T}_{i-1})$, $(\reg{S}_{i},\reg{T}_{i})$,
    $(\reg{S}_{i+1},\reg{T}_{i+1})$.
    Alternatively, $N$ can be increased to the nearest multiple of 3 without
    affecting the validity of the simulation.)
  }
  \label{fig:brick-wall}
\end{figure}
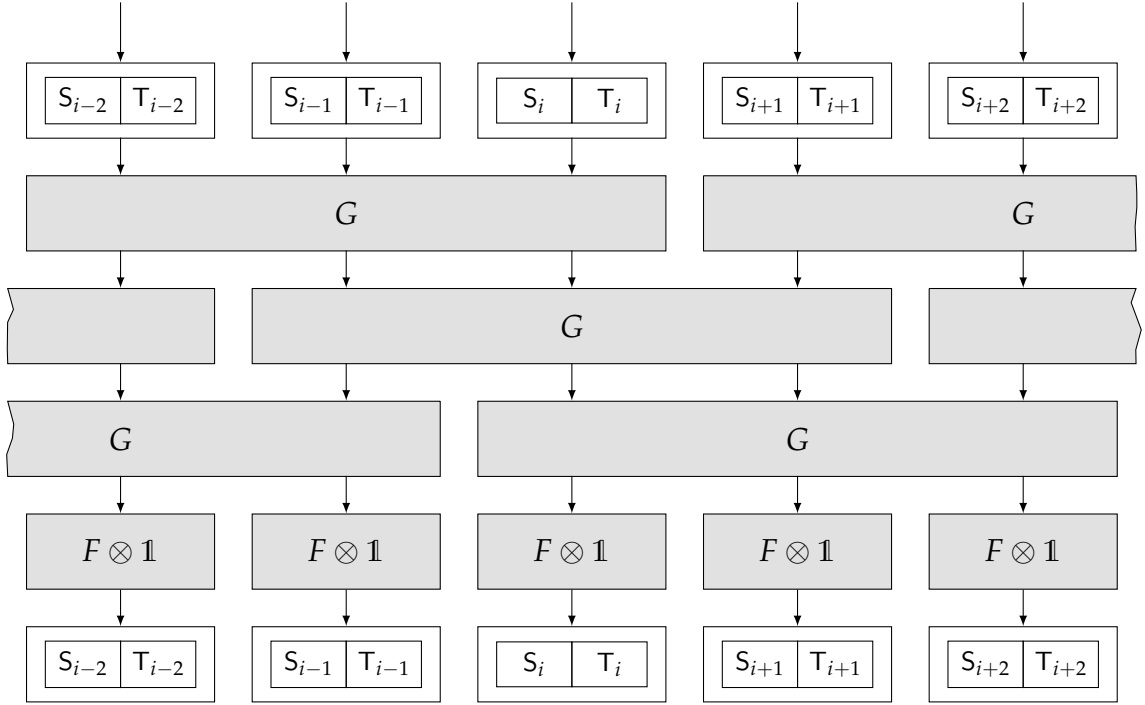

It is evident that the operator $V$ is $\reg{Y}\rightarrow\reg{X}$ causal
on the subspace $\im(\Pi)$: for an arbitrary state $\rho\in\Density(\K_N)$
satisfying $\rho = \Pi\rho\Pi$, the state
$(V \rho V^{\ast})[\reg{S}_i,\reg{T}_i]$
(i.e., the state of $(\reg{S}_i,\reg{T}_i)$ obtained by tracing out all other
registers from $V \rho V^{\ast}$) is uniquely determined by the state
$\rho[\reg{S}_{i-1},\reg{T}_{i-1},\reg{S}_i,
  \reg{T}_i,\reg{S}_{i+1},\reg{T}_{i+1}]$.
By Theorem~\ref{theorem:localizability} there must therefore exist a unitary
operator $G \in \Unitary(\X\otimes\Y)$, or equivalently
$G\in\Unitary((\S\otimes\T)^{\otimes 3})$, such that
\begin{equation}
  \label{eq:equation-defining-G}
  \Pi (G\otimes\I_{\Z}) \Pi = \Pi
  V^{\ast}((F\otimes\I)\otimes\I_{\Y\otimes\Z})V \Pi.
\end{equation}
Note that $G$ has no dependence on $i$ because the behavior of $M$ is the same
for every tape square.
The circuit suggested by Figure~\ref{fig:brick-wall}, in which $G$ is applied
to each consecutive triple $(\reg{S}_{i-1},\reg{T}_{i-1})$,
$(\reg{S}_i,\reg{T}_i)$, $(\reg{S}_{i+1},\reg{T}_{i+1})$, followed by $F$ on
each of the registers $\reg{S}_0,\ldots,\reg{S}_{N-1}$, therefore agrees with
$W$ on $\im(\Pi)$.
Given that $\im(A) \subseteq \im(\Pi)$, this operator also agrees with $W$ on
$\im(A)$.

As an aside, one may observe that the operator $V$ is not
$\reg{Y}\rightarrow\reg{X}$ causal on the entire space $\K_N$, for any
selection of $i\in\integer_N$.
For example, if $\reg{S}_j$ contains nonzero values for two or more distinct
choices of $j\in\integer_N$ (which would represent an invalid ``multi-headed''
configuration), then $V$ will act trivially on all local regions, regardless of
their states.
This explains why we have stated and proved
Theorem~\ref{theorem:localizability}, as opposed to directly making use of the
results of Arrighi, Nesme, and Werner \cite{ArrighiNW11}.

\subsection*{Behavior of the local gate $G$}

The action of the operation $G$ on standard basis states, which can be
recovered directly from the equation \eqref{eq:equation-defining-G}, is as
follows:
\begin{enumerate}
\item[1.]
  For $p_1,p_2,p_3\in\{1,\ldots,m\}$ and $a_1,a_2,a_3\in\Gamma$, $G$ acts
  trivially on the following standard basis states:
  \begin{equation}
    \begin{aligned}
      & \ket{0,a_1}\ket{p_2,a_2}\ket{0,a_3},\\
      & \ket{-p_1,a_1}\ket{0,a_2}\ket{0,a_3},\\
      & \ket{0,a_1}\ket{0,a_2}\ket{-p_3,a_3}.
    \end{aligned}
  \end{equation}
  In the first of these cases $F_i$ acts as the identity and $V$ cancels $V^\ast$, while in the other two cases all three of those operators act as the identity.
  
  Similarly, $G$ may be taken to act trivially on all choices of standard basis
  states of the form
  \begin{equation}
    \ket{q_1,a_1}\ket{q_2,a_2}\ket{q_3,a_3}
  \end{equation}
  for which two or more of the values $q_1$, $q_2$, and $q_3$ are nonzero.
  (The action of $G$ can, in fact, be chosen arbitrarily on the space spanned
  by such states, so long as this space is invariant under the action of $G$.)
  
\item[2.]
  For $p_2\in\{1,\ldots,m\}$ and $a_1,a_2,a_3\in\Gamma$, the action of $G$
  on standard basis states of the form
  $\ket{0,a_1}\ket{-p_2,a_2}\ket{0,a_3}$ is as follows:
  \begin{equation}
    \begin{aligned}
      G: \ket{0,a_1}\ket{-p_2,a_2}\ket{0,a_3}
      \mapsto \hspace{-3cm}\\[2mm]
      & \phantom{+}
      \sum_{\substack{p_1\in\{1,\ldots,m\}\\ b_1\in\Gamma}}
      \overline{\delta(p_1,b_1)[p_2,a_1,+1]}
      \ket{p_1,b_1}\ket{0,a_2}\ket{0,a_3}\\
      & +
      \sum_{\substack{p_3\in\{1,\ldots,m\}\\ b_3\in\Gamma}}
      \overline{\delta(p_3,b_3)[p_2,a_3,-1]}
      \ket{0,a_1}\ket{0,a_2}\ket{p_3,b_3}.
    \end{aligned}
  \end{equation}

\item[3.]
  For $p_1,p_3\in\{1,\ldots,m\}$ and $a_1,a_2,a_3\in\Gamma$, the action of $G$
  on standard basis states of the forms
  $\ket{p_1,a_1}\ket{0,a_2}\ket{0,a_3}$ and
  $\ket{0,a_1}\ket{0,a_2}\ket{p_3,a_3}$ is as follows:
  \begin{equation}
    \begin{aligned}
  	\label{eq:FirstEqThirdCaseForG}
      G: \ket{p_1,a_1}\ket{0,a_2}\ket{0,a_3}
      \mapsto \hspace{-4.5cm}\\[2mm]
      & \phantom{+}
      \sum_{\substack{q_2\in\{1,\ldots,m\}\\ b_1\in\Gamma}}
      \delta(p_1,a_1)[q_2,b_1,+1]
      \ket{0,b_1}\ket{-q_2,a_2}\ket{0,a_3}\\
      & +
      \sum_{\substack{
          q_1\in\{1,\ldots,m\}\\
          r_0\in\{1,\ldots,m\}\\
          b_1\in\Gamma,\; c_1\in\Gamma}}
      \delta(p_1,a_1)[r_0,c_1,-1]
      \overline{\delta(q_1,b_1)[r_0,c_1,-1]}
      \ket{q_1,b_1}\ket{0,a_2}\ket{0,a_3}.
    \end{aligned}
  \end{equation}
  and
  \begin{equation}
    \begin{aligned}
      G: \ket{0,a_1}\ket{0,a_2}\ket{p_3,a_3}
      \mapsto \hspace{-4.5cm}\\[2mm]
      & \phantom{+}
      \sum_{\substack{q_2\in\{1,\ldots,m\}\\ b_3\in\Gamma}}
      \delta(p_3,a_3)[q_2,b_3,-1]
      \ket{0,a_1}\ket{-q_2,a_2}\ket{0,b_3}\\
      & +
      \sum_{\substack{
          q_3\in\{1,\ldots,m\}\\
          r_4\in\{1,\ldots,m\}\\
          b_3\in\Gamma,\; c_3\in\Gamma}}
      \delta(p_3,a_3)[r_4,c_3,+1]
      \overline{\delta(q_3,b_3)[r_4,c_3,+1]}
      \ket{0,a_1}\ket{0,a_2}\ket{q_3,b_3}.
    \end{aligned}
  \end{equation}
\end{enumerate}

For a given quantum Turing machine $M$, the actions expressed in the second and
third items above can be simplified by making use of the conditions of
Bernstein and Vazirani required for $\delta$ to induce a global unitary
evolution.
We have described the transitions without making such simplifications to
illustrate how the transitions can simply be read off from the equation
\eqref{eq:equation-defining-G}.
Note that in the third case, one might expect from
\eqref{eq:equation-defining-G} that there would be terms where the head ends up
outside the considered range of cells.
For example, in \eqref{eq:FirstEqThirdCaseForG} we might expect to have a term
where the head first goes left with the application of $V$ and then goes left
again with the application of $V^\ast$.
However, we know from our application of Theorem \ref{theorem:localizability}
that such terms cannot exist, which means that the corresponding coefficients
that would be derived from $\delta$ must be zero for any choice of quantum
Turing machine.

\subsection*{Recapitulation}

In summary, the simulation of $M$ for $t$ steps on a given input string $x$ of
length $n\leq t$ is as follows:
\begin{enumerate}
\item[1.]
  (Pre-processing step)
  For $N = 2t+1$, initialize registers
  $(\reg{S}_0,\reg{T}_0)$, \ldots, $(\reg{S}_{N-1},\reg{T}_{N-1})$
  so that their state represents the initial configuration of $M$ on input $x$,
  running on a tape loop of length $N$.  
\item[2.]
  (Main part)
  Let $G$ be the unitary operator determined by the transition function
  $\delta$ of $M$ as described above.
  Concatenate $t$ identical copies of a circuit that first applies $G$ to every
  triple of register pairs $(\reg{S}_{i-1},\reg{T}_{i-1})$,
  $(\reg{S}_{i},\reg{T}_{i})$, $(\reg{S}_{i+1},\reg{T}_{i+1})$, for $i$ ranging
  over the set $\integer_N$, and then applies $F$ to each of the registers
  $\reg{S}_0,\ldots,\reg{S}_{N-1}$.
  The copies of $G$ can be applied in an arbitrary order, such as the 
  one suggested in Figure~\ref{fig:brick-wall} that allows for these operations
  to be parallelized.
\item[3.]
  (Post-processing step)
  Transform the standard basis states of the registers 
  $(\reg{S}_0,\reg{T}_0)$, \ldots, $(\reg{S}_{N-1},\reg{T}_{N-1})$
  that represent a configuration of $M$ on a tape loop of length $N$ into
  whatever configuration encoding is desired for the output.
\end{enumerate}

\subsection*{Complexity of the simulation}

As was already suggested, each of the registers
$\reg{S}_0,\ldots,\reg{S}_{N-1}$ and $\reg{T}_0,\ldots,\reg{T}_{N-1}$ may be
viewed as a constant-size collection of qubits, and the standard basis states
of these registers may be encoded as a binary string of an appropriate length,
and therefore the entire simulation described above may be implemented as a
quantum circuit.
Each of the registers has constant size, and can therefore be represented by a
constant number of qubits.

Let us first make the simplifying assumption that the operation $G$, which acts
on a constant number of qubits for any fixed choice of a quantum Turing machine, is available as a single quantum gate.
The total number of gates required by the main part of the simulation is
therefore $O(t^2)$, and the depth required is $O(t)$.
The pre-processing step can be performed by circuits having constant depth and
size linear in $t$, and therefore the pre-processing step and the main part of
the simulation can together be performed by quantum circuits of size $O(t^2)$
and depth $O(t)$.

The cost of the post-processing step depends on the desired form for the output
of the simulation.
For a natural choice of an encoding scheme in which each configuration
$(p,i,T)$ is described as a sequence of integers $p\in\{1,\ldots,m\}$,
$i\in\{-t,\ldots,t\}$, and $T(-t),\ldots,T(t)\in\{0,\ldots,k-1\}$, all
expressed in binary notation, the post-processing step can be performed by a
circuit with size $O(t\log(t))$ and depth $O(\log(t))$.
For a wide range of alternative encoding schemes for Turing machine
configurations, the post-processing step can be performed by circuits whose
size and depth are within the bounds $O(t^2)$ and $O(t)$ obtained for the main
part of the simulation.
For any such output form, the total number of gates required by the simulation
is therefore $O(t^2)$ and the depth is $O(t)$.

The circuits that result from the simulation described above are evidently
logarithmic-space uniformly generated (and therefore polynomial-time uniformly
generated).
To be more precise, for every quantum Turing machine $M$, there exists a
deterministic Turing machine running in logarithmic space that, on
input $1^n 0 1^t$, outputs a description of the quantum circuit that simulates
$M$ on inputs of length $n$ for $t$ steps.
This follows from the observation that these circuits all conform to the same
simple and regular pattern---the dependence on $M$ is captured entirely by the
specific choice for $G$ and the size (always considered a constant) of the
state set and alphabet of $M$.

If one is not satisfied with the assumption that $G$ is made available as
a single quantum gate from which the quantum circuits that simulate a given
quantum Turing machine can be constructed, then the cost of implementing or
approximating $G$ must be considered.
This issue is, of course, not specific to the simulation of quantum Turing
machines, but rather is a more fundamental issue---and for this reason we will
not discuss it in depth.
However, we do mention a couple of points regarding this issue that some
readers may find to be helpful:

\begin{enumerate}
\item[1.]
  The operation $G$ can be implemented exactly using a constant number
  of two-qubit gates, provided that one assumes that controlled-NOT gates and
  arbitrary single-qubit gates are available, through the method of
  \cite{BarencoBCDMSSSW95}.
  The single-qubit gates required for an exact implementation naturally depend
  on the values $\delta(p,a)[q,b,D]$ taken by the transition function of $M$.
  As $G$ is constant in size, the same complexity bounds described above
  remain valid in this case:
  the simulation requires size $O(t^2)$, depth $O(t)$, and is performed by
  logarithmic-space uniformly generated families of quantum circuits composed
  of controlled-NOT gates and a finite number of single-qubit gates
  (the selection of which depends on $M$).

\item[2.]
  If one is instead interested in a simulation with overall error $\varepsilon$
  using a fixed universal set of gates, then each $G$ must be implemented with
  accuracy on the order of $O(\varepsilon/t^2)$.
  This is possible with circuits of size polylogarithmic in $t$ and
  $1/\varepsilon$ by means of the Solovay--Kitaev theorem.
  The size and depth of the simulation in this case is as above, but multiplied
  by this polylogarithmic factor.
  This is true for arbitrary choices of the complex numbers
  $\delta(p,a)[q,b,D]$ that define the transition function of~$M$;
  the additional cost that may be incurred by difficult-to-compute numbers
  is paid only in circuit uniformity.
  That is, if computing highly accurate approximations of these numbers is
  computationally difficult, then the same will be true of computing accurate
  approximations of $G$ by a fixed gate set.
  On the other hand, the Solovay--Kitaev theorem is known to have a
  computationally efficient constructive proof \cite{KitaevSV02,DawsonN06}, and
  if accurate approximations of the complex numbers defining the transition
  function can be efficiently computed, then the same will be true of the
  circuits approximating $G$.
\end{enumerate}

\subsection*{Differences with Yao's original simulation}

Disregarding extremely minor, inconsequential differences in the way that the
simulations encode information, the key difference between the simulation
described above and Yao's original simulation is that the operation $G$ is
different in the two simulations.
Yao's simulation is similar to the one presented above in that $G$ is applied
to the three register pair neighborhood associated with each tape square,
and this operation must be applied once for each tape square in order to
simulate one step of the quantum Turing machine's computation.
Yao takes $G$ so that it directly implements the action of the Turing machine
when the tape head presence is indicated by the middle register pair:
\begin{equation}
  \begin{aligned}
    G: \ket{0,a_1}\ket{p_2,a_2}\ket{0,a_3}
    \mapsto \hspace{-4.5cm}\\[2mm]
    & \phantom{+}
    \sum_{\substack{q_2\in\{1,\ldots,m\}\\ b_2\in\Gamma}}
    \delta(p_2,a_2)[q_2,b_2,-1]
    \ket{-q_2,a_1}\ket{0,b_2}\ket{0,a_3}\\
    & +
    \sum_{\substack{q_2\in\{1,\ldots,m\}\\ b_2\in\Gamma}}
    \delta(p_2,a_2)[q_2,b_2,+1]
    \ket{0,a_1}\ket{0,b_2}\ket{-q_2,a_3}
  \end{aligned}
\end{equation}
for each $p_2\in Q$ and $a_1,a_2,a_3\in\Gamma$.
(Yao actually does this for quantum Turing machines allowing for stationary
tape heads, but this is the form for quantum Turing machines that disallow
for stationary tape heads.)

The operation $G$ is then further constrained so that it acts trivially on a certain subspace. This requires that one considers the orthogonality relations induced by the unitary global evolution of $M$, specifically among the states obtained when
$M$ is run on configurations in which the tape head has distance one or two from the cell represented by the middle triple upon which $G$ acts. Yao does not explicitly describe $G$, but observes that it may be obtained
through basic linear algebra.

As analyzed by Yao, this leads to the correctness of a cascading construction,
where the instances of $G$ are executed from left to right.
However, is not difficult to see that if one wishes to parallelize this
construction, it can be done so through a minor symmetry-inducing change in the
definition of the subspace that $G$ acts trivially on.
After that, a given instance of $G$ will commute with those instances applied
to the two overlapping neighborhoods consisting of three register pairs. 

\subsection*{Simulating variants of quantum Turing machines}

The simulation method described above can be applied to variants of quantum
Turing machines that exhibit a local causal behavior similar to ordinary
(one-dimensional tape) quantum Turing machines.
A very simple example, which we have already noted was considered in Yao's
original paper, is that of a quantum Turing machine whose tape head may remain
stationary, so that its transition function takes the form
\begin{equation}
  \delta: Q\times \Gamma \rightarrow \complex^{Q\times\Gamma\times\{-1,0,+1\}}.
\end{equation}
There is essentially no difference in the analysis of this case from the one
presented above, except that an explicit specification of the operation $G$,
which is easily obtained from the equation \eqref{eq:equation-defining-G},
may be slightly more complicated than the one described above.
This lack of substantial differences is a consequence of the fact that quantum
Turing machines allowing for stationary heads have precisely the same causal
structure that was required to invoke Theorem~\ref{theorem:localizability}.

A different example that better illustrates the flexibility of the simulation
method we have described is that of quantum Turing machines having
multi-dimensional tapes.\footnote{The word \emph{tape} is perhaps a poor
  choice of a word to describe a multi-dimensional storage medium, but little
  would be gained in introducing a different term for such an object.}
It is not necessarily our intention to advocate further study of this arguably
contrived quantum Turing machine variant---the discussion that follows is meant
only to support the claim that our simulation and its analysis extend without
complications to models other than the standard quantum Turing machine model
with a single one-dimensional tape.
We will consider just two-dimensional tapes in the interest of simplicity, but
it will be apparent that the discussion may be extended to tapes of any
constant dimension.

A natural way to define a quantum Turing machine with a two-dimensional
tape is by a transition function of the form
\begin{equation}
  \delta: Q\times \Gamma \rightarrow
  \complex^{Q\times\Gamma\times\{-1,0,+1\}\times\{-1,0,+1\}}.
\end{equation}
The interpretation of such a function is that the complex number
\begin{equation}
  \delta(p,a)[q,b,D_1,D_2]
\end{equation}
indicates the amplitude with which the machine will, when in state $p$ and
scanning a tape square containing the symbol $a$, change state to $q$,
overwrite the symbol in the square being scanned with $b$, and move its tape
head in the direction $(D_1,D_2)$ on the tape.
(For example, $(-1,+1)$ indicates a diagonal tape head movement, up and to the
left.)
A configuration of a two-dimensional tape Turing machine having state set $Q$
and tape alphabet $\Gamma$ is represented by a triple $(p,(i,j),T)$, where
$p\in Q$ is a state, $(i,j) \in \integer\times\integer$ is a pair of integers
representing the tape head location, and
$T:\integer\times\integer\rightarrow\Gamma$ is a function with finite support
that describes the contents of the two-dimensional tape.
To simulate such a quantum Turing $M$ for $t$ steps on an input of length
$n\leq t$, it may be imagined that the machine runs on a tape in the form of a
torus indexed by $\integer_N\times\integer_N$, for $N = 2t+1$, which is the
natural two-dimensional analogue of a loop in one dimension.

The quantum circuit simulation described earlier naturally extends to this
situation, with one register pair $(\reg{S}_{i,j},\reg{T}_{i,j})$ being defined
for each tape square.
The roles played by these registers are similar to before:
$\reg{S}_{i,j}$ indicates whether or not the tape head is present at square
$(i,j)$, and the state (active or inactive) if it is, and $\reg{T}_{i,j}$
represents the tape symbol stored in tape square $(i,j)$.
The transition function of $M$ is assumed to define a unitary evolution, which
in turn defines a unitary operator $V$ acting on the state space $\K_N$
(corresponding to the $N^2$ register pairs just described) in a similar manner
to the one-dimensional case.
This unitary operator acts trivially on the subspace of $\K_N$ orthogonal to
the standard basis states in which precisely one register $\reg{S}_{i,j}$
contains a positive value and none contain negative values.
The operator $F$, which may be applied to any register $\reg{S}_{i,j}$, and
the projection $\Pi$ onto the subspace of $\K_N$ spanned by standard basis
states in which exactly one register $\reg{S}_{i,j}$ contains a nonzero value,
are defined in the same way as in the one-dimensional case.

In the one-dimensional case, the operation $G$ acts on three register pairs,
but in the two-dimensional case $G$ acts on nine register pairs.
Specifically, for each index $(i,j)$, the corresponding $G$ operation acts on
the register pairs having indices in the set
\begin{equation}
  \bigl\{(i',j')\,:\,\abs{i-i'}\leq 1,\; \abs{j-j'}\leq 1\bigr\},
\end{equation}
which is equivalent to the set containing $(i,j)$ and its 8 nearest neighbors
on the torus $\integer_N\times\integer_N$.
The register pairs indexed by elements in this set may be collected into
compound registers as $\reg{X} = (\reg{S}_{i,j},\reg{T}_{i,j})$ and
\begin{equation}
  \reg{Y} =
  \left(
  \begin{array}{ccc}
    (\reg{S}_{i-1,j-1},\reg{T}_{i-1,j-1}) & 
    (\reg{S}_{i,j-1},\reg{T}_{i,j-1}) &
    (\reg{S}_{i+1,j-1},\reg{T}_{i+1,j-1}) \\[2mm]
    (\reg{S}_{i-1,j},\reg{T}_{i-1,j}) & 
     &
    (\reg{S}_{i+1,j},\reg{T}_{i+1,j}) \\[2mm]
    (\reg{S}_{i-1,j+1},\reg{T}_{i-1,j+1}) & 
    (\reg{S}_{i,j+1},\reg{T}_{i,j+1}) &
    (\reg{S}_{i+1,j+1},\reg{T}_{i+1,j+1}).
  \end{array}
  \right)
\end{equation}
Taking $\reg{Z}$ to include all of the remaining register pairs aside from the
ones included in $\reg{X}$ and $\reg{Y}$, one may apply
Theorem~\ref{theorem:localizability} to conclude that there exists a unitary
operator $G$ acting on $(\reg{X},\reg{Y})$ such that
\begin{equation}
  \Pi (G\otimes\I_{\Z}) \Pi = \Pi
  V^{\ast}((F\otimes\I)\otimes\I_{\Y\otimes\Z})V \Pi.
\end{equation}
For each possible standard basis state of the register pairs comprising
$\reg{X}$ and $\reg{Y}$, the action of $G$ may effectively be read off from
this equation.
One observes the following:
\begin{enumerate}
\item[1.]
  For standard basis states of these registers for which none of the
  registers $\reg{S}_{i,j}$ included in $\reg{X}$ and $\reg{Y}$ is nonzero, the
  action of $G$ will necessarily be trivial.
\item[2.]
  For standard basis states of these registers for which two or more of the
  registers $\reg{S}_{i,j}$ included in $\reg{X}$ and $\reg{Y}$ are nonzero,
  the action of $G$ may be taken to be trivial.
  (The above equation implies that the subspace spanned by such states must be
  invariant under the action of $G$, but taking $G$ to be the identity on this
  space is the simplest choice.)
\item[3.]
  For all remaining standard basis states, one may assume that every register
  of $\reg{Z}$ is in the standard basis state $(0,0)$ for simplicity, although
  this choice will have no influence on the action that is recovered for $G$.
  By Theorem~\ref{theorem:localizability}, the action of the operator on the
  right-hand side of the above equation will uniquely specify the action of $G$
  on the chosen standard basis state; and moreover the same theorem implies
  that $G$ will be unitary.
\end{enumerate}

As in the one-dimensional case, the action of $G$ is independent of the choice
of $(i,j)$, and the operators obtained by applying $G$ to two neighborhoods
corresponding to distinct choices of $(i,j)$ will necessarily commute, even
when the neighborhoods overlap.
A simulation of $M$ is obtained as before, by alternating between the
application of $G$ to the nine register pairs corresponding to the
neighborhoods of $(\reg{S}_{i,j},\reg{T}_{i,j})$ for every pair $(i,j)$ with
the application of $F$ to every register $\reg{S}_{i,j}$.
This time the size of the resulting circuits is $O(t^3)$ rather than $O(t^2)$,
while the depth remains linear in $t$.

Yao's original simulation can also be extended to quantum Turing machines with
multi-dimensional tapes, although once again one is required to perform
computations based on simple linear algebra---which become increasingly tedious
as the dimension of the tape grows---to obtain a description of $G$.

\bibliographystyle{alpha}
\bibliography{QTM}

\end{document}